\documentclass[let tersize,journal]{IEEEtran}
\usepackage{amsmath}
\usepackage{cite}
\usepackage{amssymb}
\usepackage{amsfonts}
\usepackage{graphicx}
\usepackage{textcomp}
\usepackage{xcolor}
\usepackage{bbm}
\usepackage{caption}

\usepackage{url}
\usepackage{colortbl}
\usepackage[normalem]{ulem}
\useunder{\uline}{\ul}{}

\useunder{\uline}{\ul}{}
\usepackage{mwe}
\usepackage{comment}
\usepackage[export]{adjustbox}
\usepackage{amsthm}

\def\cA{{\cal A}}

\def\bW{{\mathbf W}}

\def\by{{\mathbf y}}
\def\bx{{\mathbf x}}

\def\balpha{{\boldsymbol{\alpha}}}

\def\bomega{{\boldsymbol{\omega}}}

\def\bA{{\mathbf A}}
\def\bB{{\mathbf B}}
\def\bS{{\mathbf S}}
\def\ba{{\mathbf{a}}}

\newcommand{\projop}{{\mathcal{P}}}

\def\stochmeas{{\nu}}
\def\detmeas{{\mu}}

\def\genmeas{{\mu}}

\newtheorem{proposition}{Proposition}

\newtheorem{remark}{Remark}
\def\minwrt[#1]{\underset{#1}{\mathrm{minimize }}}

\def\mminwrt[#1]{\underset{#1}{\mathrm{min }}}

\def\maxwrt[#1]{\underset{#1}{\mathrm{maximize }}}
\def\argminwrt[#1]{\underset{#1}{\mathrm{arg min }}}

\usepackage{float}

\newcommand{\otdist}{\mathcal{S}}
\newcommand{\otdistsp}{\bS_c}

\newcommand{\harmotdist}{S_{\mathrm{h}}}
\newcommand{\groundcost}{c}

\newcommand{\bPsi}{\mathbf{\Psi}}

\newcommand{\bhr}{\mathbf{\hat{r}}}

\newcommand{\br}{\mathbf{r}}

\newcommand{\bc}{\mathbf{c}}

\newcommand{\bu}{\mathbf{u}}
\newcommand{\bv}{\mathbf{v}}

\newcommand{\bC}{\mathbf{C}}
\newcommand{\bK}{\mathbf{K}}
\newcommand{\bM}{\mathbf{M}}
\newcommand{\bQ}{\mathbf{Q}}

\newcommand{\blambda}{\boldsymbol{\lambda}}
\newcommand{\bmu}{\boldsymbol{\mu}}
\newcommand{\bxi}{\boldsymbol{\xi}}
\newcommand{\onevec}{\mathbf{1}}
\newcommand{\bnu}{\boldsymbol{\nu}}
\newcommand{\RR}{\mathbb{R}}
\newcommand{\CC}{\mathbb{C}}
\newcommand{\NN}{\mathbb{N}}
\newcommand{\EE}{\mathbb{E}}
\newcommand{\ZZ}{\mathbb{Z}}

\newcommand{\pmeas}{{\mathcal{M}_+}}
\newcommand{\cmeas}{{\mathbb{C}\mathcal{M}}}

\newcommand{\norm}[1]{\left\lVert#1\right\rVert}
\newcommand{\abs}[1]{\left|#1\right|}

 \newcommand{\indset}[1]{[\![#1]\!]}% integers up to and including

 \newcommand{\TT}{\mathbb{T}} % unit circle 

 \newcommand{\cper}{{\mathcal{C}_{\mathrm{per}}}} % periodic functions 

\newcommand{\genreg}{\mathcal{R}}

\newcommand{\stochreg}{\mathcal{R}^{\mathrm{stoch}}}

\newcommand{\detreg}{\mathcal{R}^{\mathrm{det}}}

\newcommand{\Tlag}{\mathcal{T}}

\newcommand{\stepsize}{\gamma}
\newcommand{\wrightOmega}{\mathcal{W}}
\begin{document}

%
% Title.
% ------%
\title{Inverse harmonic clustering for multi-pitch \\estimation: an optimal transport approach\thanks{The authors are with the Department of Information and Communications Engineering, Aalto University, Finland. Emails: \texttt{firstname.lastname@aalto.fi}}}

\author{\IEEEauthorblockN{Anton Bj{\"o}rkman}
and
\IEEEauthorblockN{Filip Elvander}}
\maketitle
\begin{abstract}
%
%%%%%%%%%%%%%%%%%%%%%%%
%
In this work, we consider the problem of multi-pitch estimation, i.e., identifying super-imposed truncated harmonic series from noisy measurements. We phrase this as recovering a harmonically-structured measure on the unit circle, where the structure is enforced using regularizers based on optimal transport theory. In the resulting framework, a signal’s spectral content is simultaneously inferred and assigned, or transported, to a small set of harmonic series defined by their corresponding fundamental frequencies. In contrast to existing methods from the compressed sensing paradigm, the proposed framework decouples regularization and dictionary design and mitigates coherency problems. As a direct consequence, this also introduces robustness to the phenomenon of inharmonicity. 
From this framework, we derive two estimation methods, one for stochastic and one for deterministic signals, and propose efficient numerical algorithms implementing them. In numerical studies on both synthetic and real data, the proposed methods are shown to achieve better estimation performance as compared to other methods from statistical signal processing literature. Furthermore, they perform comparably or better than network-based methods, except when the latter are specially trained on the data-type considered and are given access to considerably more data during inference.

%
%%%%%%%%%%%%%%%%%%%%%%%
%
\end{abstract}
\begin{IEEEkeywords}
Multi-pitch estimation, fundamental frequency, inharmonicity, optimal transport, spectral estimation
\end{IEEEkeywords}
\section{Introduction}
Fundamental frequency, or pitch, estimation is a common task in a variety of applications within signal processing, including speech processing \cite{norholm2016instantaneous}, music audio signal processing \cite{Muller2011, duan2011soundprism}, and biomedical modeling \cite{4636708}. The task can be divided into two classes: single-pitch estimation, and multi-pitch estimation, with the latter being significantly more challenging, as well as less studied. Single-pitch estimation aims to identifying the fundamental frequency of one signal that consists of a single harmonic series, as seen in cases like a single instrument or for human speech with only one active speaker. In contrast, multi-pitch estimation arises naturally when there are multiple instruments playing or when there are several speakers. For the single-pitch problem, there are several methods available. Examples of these include YIN \cite{Cheveigne2002YIN}, RAPT \cite{talkin1995robust} and SWIPE \cite{camacho2008sawtooth}, where both YIN and RAPT are based on the autocorrelation of the signal. SWIPE instead attempts to find the best match for the fundamental frequency of a sawtooth waveform that best describes the spectrum of the input signal.
For the multi-pitch problem (for a general overview, see \cite{christensen2022multi}), some recent methods are parametric methods based on dictionaries whose atoms correspond to harmonic series, such as PEBS \cite{adalbjornsson2013PEBS}, and the signal components are identified by minimizing sparsity-promoting optimization criteria, often expressed as a regularized least-squares problem. An extension of PEBS was introduced in \cite{ADALBJORNSSON2015236} as PEBS$_2$TV, adding a total variation term to deal with subharmonic errors in the model. An additional method, PEBSI-lite, was introduced in \cite{ELVANDER201656}, improving upon PEBS$_2$TV by better approximating one of the penalty terms as well as introducing self-regularization and an adaptive dictionary construction. Another method based on the same type of dictionary is PE-BSBL-Cluster (herein denoted PE-BSBL-C) \cite{shi2018multipitch}, which instead uses Bayesian learning with a block sparse prior, as well as intra-block clustering to deal with subharmonic errors. Here, the estimation is based on assumptions of the harmonic structure of the signal, such as an added penalty if the first harmonic is missing, or if nearby harmonics are missing. One of the drawbacks of these methods is that the dictionary atoms become highly coherent because of the overlap of the harmonics of different candidate pitches. Furthermore, these methods assume perfect harmonicity among each pitch, making them susceptible to errors in the face of inharmonicity, i.e., when individual sinusoids deviate from the harmonic model, as is the case, for example, with string instruments \cite{fletcher1964normal, fletcher2012physics, rasch1985string} and to some extent in human speech \cite{fernandes2018harmonic}.

Recently, pitch estimators based on deep learning have been proposed, showing promising results both for the single-pitch and multi-pitch problems. Estimators for the single-pitch problem include CREPE \cite{kim2018crepe}, Harmof0 \cite{wei2022harmof0} and SPICE \cite{gfeller2020spice}, while examples for the multi-pitch problem are DeepSalience \cite{bittner2017deep} as well as three models from \cite{cuesta2020multiple}, building upon the former. While these methods have proven effective for pitch estimation, they inherently impose certain limitations on the input. For example, if a model is trained on a specific type of data, such as speech, it may not be effectively detect pitches in other signal types like music. Moreover, parameters such as sampling rate, frame length, frequency range, and hop size are often enforced by the structure of the model itself, or required to match that of the training data. As a result, if one would want to change these parameters, the model would either need to be retrained, or a pre-processing step, such as re-sampling the signal, would need to be applied to the signal, as is commonly done for these estimators. 

In similarity to the previously mentioned dictionary-based multi-pitch estimators, our proposed work uses a similar dictionary in order to form our estimates. However, in contrast to the other dictionary-based methods, we in this work propose to both alleviate the need of pitch-based dictionaries by formulating the multi-pitch estimation through a dictionary, but instead of having the atoms correspond to complete harmonic series, the atoms in our dictionaries correspond to Fourier basis functions, thus removing the harmonic overlap between atoms. In particular, the aim is to infer the spectral content of the observed signal, i.e., the power (or amplitudes) and frequencies of individual sinusoids, while simultaneously assigning these into a small number of groups, corresponding to pitches. This grouping is performed on the basis of the Monge-Kantorovich problem of optimal transport (OT) \cite{villani2009optimal}. OT-based problems have previously been used for spectral estimation (see, e.g., \cite{georgiou2008metrics, elvander2018interpolation,elvander2020multi,Haasler2024}), as well as in the context of pitch estimation \cite{flamary2016optimal,elvander2017using, elvander2023estimating, Elvander2020HarmonicDef}. However, to the best of the authors' knowledge, our previous work \cite{bjorkman2025robust} was the first work to propose an OT-based estimator for solving the multi-pitch estimation as an \textit{inverse} problem, where only noisy observations of the signal waveform are known, of which the work proposed herein is expanded upon.

In this paper, we aim to build upon our prior work in \cite{bjorkman2025robust}. This work differs from our previous work in the following aspects: firstly, we define the problem in a more general form, which not only allows us to more precisely articulate the specific problem we aim to solve, but also allows us to return to it later: we first solve the pitch estimation problem, return to the general form, and then re-estimate the frequency structure of the signal. Secondly, we introduce a new estimation formulation, based on the auto-covariance of the measurements. Thirdly, we conduct more thorough experiments than in our previous work, here including real data as well as a Monte Carlo type study. We also provide a comparison with deep network-based methods, which was not included in our previous work, wherever applicable.

\subsection{Notation}
Unless explicitly stated otherwise, we denote scalars, column vectors, and matrices by lower-case, lower-case bold, and upper-case bold letters such as $a$, $\ba$, and $\bA$. Let $(\cdot)^T$, $(\cdot)^H$, and $\overline{(\cdot)}$ denote transpose, Hermitian transpose, and complex conjugation, respectively. For $F \in \NN$, we use the short-hand $\indset{F} \triangleq \{ 1,2,\ldots, F\}$. We let $\TT \triangleq [-\pi,\pi)$, which we identify with the unit circle. Furthermore, let $\cper(\TT)$ be the set of continuous, complex-valued, 2$\pi$-periodic functions on $\TT$. We let $\cmeas(\TT)$ denote the set of linear bounded functionals on $\cper(\TT)$, i.e., the dual space of $\cper(\TT)$. In particular, this is the set of complex-valued measures on $\TT$ with finite total variation. Let $\pmeas(\TT)$ denote the subset of $\cmeas$ consisting of real, non-negative measures. For $f \in \cper(\TT)$ and $\mu \in \cmeas(\TT)$, we use the notation $\langle \mu, f\rangle $ for the application of the functional $\mu$ on $f$, i.e.,
\begin{align*}
    \langle \mu, f\rangle = \int_\TT f(\omega) d\mu(\omega).
\end{align*}
For Hilbert spaces, we use $\langle \cdot, \cdot\rangle$ to denote the standard inner product. In particular, for two matrices $\bA, \bB \in \CC^{m \times n}$, $\langle \bA, \bB \rangle = \text{trace}\left(\bA \bB^H\right)$. For scalars $a$, $\abs{a}$ denotes the absolute value, whereas for a vector $\ba$, the absolute value is applied elementwise, as in $\abs{\ba}$. Furthermore, for $\mu \in \cmeas(\TT)$, $\abs{\mu}$ denotes the variation, i.e., for any $\mathcal{E} \subset \TT$
\begin{align*}
    \abs{\mu}(\mathcal{E}) = \sup \sum_{p} \abs{\mu(\mathcal{E}_p)},
\end{align*}
where the supremum is over all partitions of $\mathcal{E}$. Note that $\abs{\mu} \in \pmeas(\TT)$, and if $\mu \in \pmeas(\TT)$, then $\abs{\mu} = \mu$.

\section{Signal model}
Consider the discrete-time, complex-valued signal\footnote{For generality, we here consider the complex-valued representation. However, this can easily be formed as the discrete-time analytical version of a real-valued signal \cite{marple1999computing}.} $y$ on the form
\begin{align} \label{eq:noisy_signal}
    y_t = x_t + v_t,\; t = 0,1,2,\ldots
\end{align}
where $x$ is the signal of interest, and $v$ is an additive noise term, assumed to be circularly symmetric white Gaussian noise with some variance $\sigma_v^2$, i.e., $v_t \sim \CC\mathcal{N}(0,\sigma_v^2)$ and independent for different $t$. We here assume that $x$ is an inharmonic multi-pitch signal. That is, $x$ can be decomposed into $K \in \NN$ components according to
\begin{align}\label{eq:almost_harmonic_signal}
    x_t = \sum_{k = 1}^K x_t^{(k)},
\end{align}
where each $x^{(k)}$ is an approximately periodic, or \textit{inharmonic}, waveform. In particular, for some harmonic order $L^{(k)} \in \NN$, each component is of the form
\begin{align} \label{eq:component_waveform}
    x_t^{(k)} = \sum_{\ell = 1}^{L^{(k)}} \alpha_\ell^{(k)} e^{i\omega_\ell^{(k)} t},
\end{align}
where $\alpha_\ell^{(k)} \in \CC$, $\ell \in \indset{L^{(k)}}$, are complex amplitudes, and the frequencies $\omega_\ell^{(k)} \in \TT$ form an approximately harmonic set:
\begin{align} \label{eq:inharm_set}
    \omega_\ell^{(k)} \approx \ell \omega_0^{(k)}
\end{align}
for some $\omega_0^{(k)} \in \TT$. The approximate equality in \eqref{eq:inharm_set} should here be interpreted as the existence of so-called inharmonicity parameters $\Delta_\ell^{(k)}$, with $\abs{\Delta_\ell^{(k)}} \ll \omega_0^{(k)}$, such that
\begin{align*}
    \omega_\ell^{(k)} = \ell \omega_0^{(k)} + \Delta_\ell^{(k)}.
\end{align*}
For an in-depth discussion on how to stringently define a concept of "fundamental frequency" $\omega_0^{(k)}$ for inharmonic signals, see \cite{Elvander2020HarmonicDef}. Then, given a finite set of noisy samples of the signal waveform $x$, i.e., observations $y_t$ from \eqref{eq:noisy_signal} for $t = 0,1,\ldots,N-1$ for a finite $N \in \NN$, we wish to find a decomposition \eqref{eq:almost_harmonic_signal} and to estimate the set of fundamental frequencies $\{ \omega_0^{(k)} \}_{k=1}^K$. We refer to this problem as inharmonic multi-pitch estimation, where each component $x^{(k)}$ is referred to as a pitch.
\subsection{Prior art and motivation}
In the perfectly harmonic case, i.e., when the relation \eqref{eq:inharm_set} holds with equality, the state-of-the-art methods for multi-pitch estimation in the statistical signal processing literature revolve around the idea of parsimonious signal representations, and in particular the concept of compressed sensing. In broad terms, many methods, such as the methods in \cite{adalbjornsson2013PEBS, ADALBJORNSSON2015236, ELVANDER201656}, solve problems of the form
\begin{align} \label{eq:generic_estimation}
    \minwrt[\balpha] \quad\frac{1}{2} \norm{ \by - \bA \balpha }_2^2 + \genreg(\balpha),
\end{align}
where $\by = \begin{bmatrix} y_0 & y_1 & \ldots & y_{N-1} \end{bmatrix}^T$ is the collection of signal samples, $\bA$ is finite dictionary used for approximating the signal, and $\balpha$ is the corresponding set of coefficients. Furthermore, $\genreg$ is a regularization function that, together with a properly chosen structure of $\bA$, is designed to promote a sparse pitch-structure on the estimate $\bA \balpha$. Typically, the "atoms" of the dictionary $\bA$ are blocks of column vectors corresponding to a truncated harmonic series. That is, for such a block $\bA^{(k)} \in \CC^{N \times L^{(k)}}$ and a corresponding coefficient vector $\balpha^{(k)} \in \CC^{L^{(k)}}$, the vector $\bA^{(k)} \balpha^{(k)}$ represents samples from (perfectly harmonic) signals on the form \eqref{eq:almost_harmonic_signal}. The regularizer $\genreg$ is then designed as to promote a block-sparse structure on the full coefficient vector $\balpha$. The estimated fundamental frequencies $\{\omega_0^{(k)}\}_{k = 1}^K$ are simply found by identifying the corresponding non-zero blocks of $\balpha$. Although the method presented in \cite{shi2018multipitch} does not follow the structure of \eqref{eq:generic_estimation}, the same type of dictionary is used, with sparsity enforced in a similar manner, but through the use of a Bayesian framework. 
%%%%%
\begin{figure}[t]
    \centering
    \includegraphics[width = \linewidth]{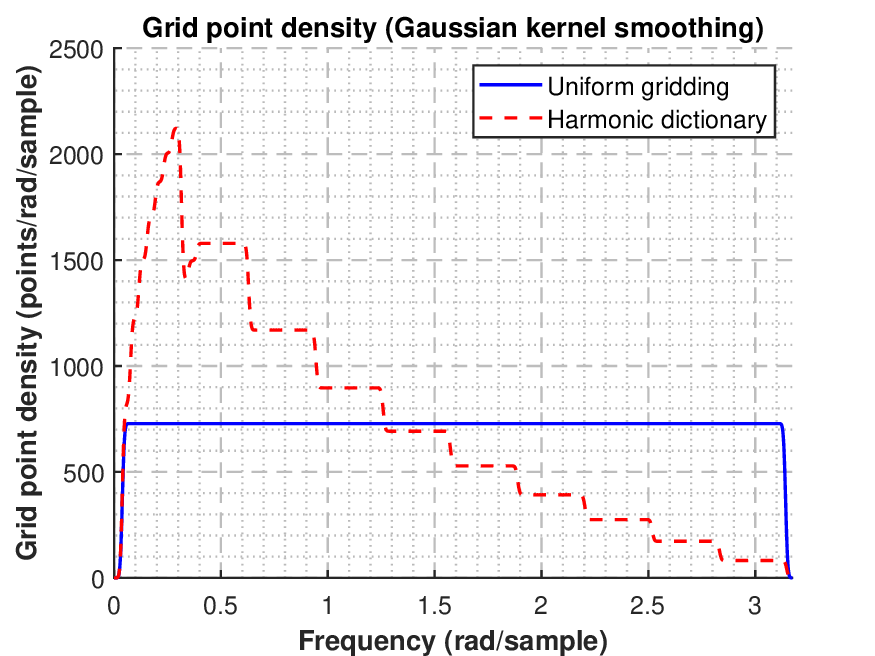}
    \caption{ Example of the frequency grid point density. The case of uniform gridding is shown in blue, while the frequency grid used for the harmonic dictionaries is shown in red. Both grids have the same number of total grid points.
    }  
    \label{fig:ngridp}\vspace{-1mm}
\end{figure}
%%%%%
%%%%%
Although these methods can yield high-accuracy estimates, they are susceptible to the well-known off-grid effect \cite{ChiSPC11_59}. This may be countered by selecting a fine-enough set of atoms of $\bA$, i.e., candidate pitch-blocks $\bA^{(k)}$. However, achieving a desired coverage of the interval\footnote{It may be noted the sampling of the frequency axis implied by the full set of harmonics in $\bA$ is highly non-uniform.} $\TT$ leads to a large and highly coherent dictionary $\bA$. This affects both estimation performance, as well as computational burden when solving \eqref{eq:generic_estimation}. Furthermore, this rigid structure does not allow for deviations from the perfectly harmonic assumption, such as the inharmonicities considered in this work. It may be noted that also grid-free and grid-less methods have been proposed, building on the ideas of atomic norm denoising \cite{ChandrasekaranRPW12_12,BhaskarTR13_61}. These, however, require the solution of large semi-definite programs \cite{JensenV17_icassp} or rely on non-convex heuristics \cite{swardlj18_26}. Furthermore, neither these accommodate deviations from perfect harmonicity.

In contrast, we in this work propose to phrase multi-pitch estimation as the problem of recovering a measure from noisy indirect measurements, where we design regularizing functions $\genreg$ imposing a pitch structure on the sought measure. As will later be shown, by leveraging the concept of optimal transport, this allows us to decouple the measurement mechanism, i.e., the mapping taking the signal "coefficients" to the observations, from the regularizer. This is in stark contrast with the above mentioned methods, in which both measurement and regularization are inherently linked to the structure of the dictionary $\bA$. As an illustration, Figure~\ref{fig:ngridp} shows the sampling of the frequency axis resulting from dictionaries with harmonic series as atoms. As can be seen, this results in a highly non-uniform frequency gridding. In addition to this, it should be emphasized that this type of dictionary will also result in duplicated grid points due to harmonic overlap. As reference, the uniform sampling of a Fourier vector dictionary, as used by the methods proposed herein, is shown. 

As noted, we will construct our regularizer based on the concept of optimal transport. Before doing this, we will connect the signal $x$ to complex-valued measures on $\TT$, i.e., elements of $\cmeas(\TT)$. This will allow us to talk about notions of "mass" for $x$, which naturally leads to transport formulations. We will do this in two different settings: one stochastic and one deterministic. Although related, the two different views ultimately lead to different regularizers and estimators.
\subsection{Stochastic model}
Assume that the complex amplitudes $\alpha_\ell^{(k)}$ in \eqref{eq:component_waveform} are random variables parametrized as $\alpha_\ell^{(k)} = \sigma_\ell^{(k)}z_\ell^{(k)}$, where $\sigma_\ell >0$ are deterministic, and $z_\ell$ are independent, zero-mean, circularly symmetric random variables with unit variance, i.e.,
\begin{align*}
    \EE\left( z_\ell^{(k)} \right) = 0, \; \EE\left( \abs{z_\ell^{(k)}}^2\right) = 1,
\end{align*}
where $\EE(\cdot)$ is the expectation operator. Then, the signal $x$ is a wide-sense stationary stochastic process with covariance function $r: \ZZ \to \CC$ defined as
% %
%
\begin{align*}
    r(\tau) \triangleq \EE\left( x_t \overline{x_{t-\tau}} \right) = \sum_{k = 1}^{K}\sum_{\ell = 1}^{L^{(k)}}\ \sigma_\ell^2 e^{i\omega_\ell^{(k)} \tau}, \; \tau \in \ZZ.
\end{align*}
With this, we can associate $x$, or $r$, with a non-negative measure $\stochmeas$ on $\TT$, i.e., $\stochmeas \in \pmeas(\TT) \subset \cmeas(\TT)$, in the sense that
\begin{align}
    r(\tau) = \langle \stochmeas, e^{i(\cdot) \tau} \rangle = \int_{\TT} e^{i\omega \tau} d\stochmeas(\omega), \; \tau \in \ZZ.
    \label{eq:val_of_linear_functional_stoch}
\end{align}
That is, $\nu$ generates the covariance function by acting on functions $f_\tau(\omega) = e^{i\omega \tau}$, indexed by $\tau$. Specifically,
\begin{align}
\label{eq:meas_harm}
    \stochmeas = \sum_{k=1}^K\sum_{\ell = 1}^{L^{(k)}} \left(\sigma_\ell^{(k)}\right)^2\delta_{\omega_\ell^{(k)}},
\end{align}
where $\delta_{q}$ is the Dirac measure at $q \in \TT$. This is simply stating the Wiener-Khinchin theorem, and $\stochmeas$ is the power spectrum of $x$. Then, if we consider a finite set of observations at lags $\tau = 0,1,\ldots,\Tlag-1$, with $\Tlag \in \NN$, for the covariance function $r$, we can define the linear operator $\cA : \cmeas(\TT) \to \CC^\Tlag$ such that $\left[ \cA(\stochmeas) \right]_\tau = \langle \stochmeas, e^{i(\cdot) \tau} \rangle$, i.e., 
\begin{align} \label{eq:cov_relation}
    \br = \cA(\stochmeas),
\end{align}
where $\br = \begin{bmatrix} r(0) & r(1) & \ldots  & r(\Tlag -1)\end{bmatrix}^T$. Thus, if we have access to $\br$, or an estimate there-of, one could imagine retrieving $\stochmeas$, and thus the frequency content of $x$, by inverting the relation \eqref{eq:cov_relation}. If the observation $\br$ is noise-free and if $\Tlag > \sum_{k=1}^K L^{(k)}$, this is indeed possible by the Caratheodory-Fejér theorem \cite{Grenander1958}, as is done for instance in the MUSIC algorithm \cite{schmidt1986multiple}. In practice, $\br$ is subject to noise, and we need to regularize the relation in order to allow for inversion.

Note here that the spectrum $\stochmeas$ can be interpreted as a distribution of mass over $\TT$. Furthermore, $\stochmeas$ can be decomposed into $K$ approximately harmonic component measures $\stochmeas^{(k)} \in \pmeas(\TT)$ as
\begin{align} \label{eq:stoch_decomposition}
    \stochmeas = \sum_{k=1}^K \stochmeas^{(k)}, \text{ with } \stochmeas^{(k)} = \sum_{\ell = 1}^{L^{(k)}} \left(\sigma_\ell^{(k)}\right)^2\delta_{\omega_\ell^{(k)}}.
\end{align}
Note that from $\stochmeas^{(k)}$, the frequency content of the corresponding $x^{(k)}$ is directly available. This additive decomposition and mass distribution interpretation will allow us to leverage optimal transport theory in order to retrieve a "multi-pitch structured" $\stochmeas$.

\subsection{Deterministic model}
Herein, we take the alternative view that the complex amplitudes $\alpha_l^{(k)}$ are deterministic constants. Although we lose the concepts of covariance function and power spectrum, we can still associate the signal $x$ with an element of $\cmeas(\TT)$. In fact, this measure $\detmeas \in \cmeas(\TT)$ now generates the signal samples themselves, i.e.,
\begin{align*}
    x_t = \langle \detmeas, e^{i(\cdot) t} \rangle = \int_{\TT} e^{i\omega t} d\detmeas(\omega), \; t \in \NN,
\end{align*}
by acting on the functions $f_t(\omega) = e^{i\omega t}$, indexed by the time $t$. Here, $\detmeas$ is given by
\begin{align} \label{eq:deterministic_measure}
    \detmeas =  \sum_{k=1}^K\sum_{\ell = 1}^{L^{(k)}} \alpha_\ell^{(k)}\delta_{\omega_\ell^{(k)}},
\end{align}
where, similar to the stochastic case, we note that $\detmeas$ may be decomposed into $\detmeas^{(k)} \in \cmeas(\TT)$ according to
\begin{align*}
    \detmeas =  \sum_{k=1}^K \detmeas^{(k)}, \text{ with } \detmeas^{(k)} = \sum_{\ell = 1}^{L^{(k)}} \alpha_\ell^{(k)}\delta_{\omega_\ell^{(k)}}.
\end{align*}
If we have access to $N \in \NN$ observations of $x_t$, $t = 0,1,\ldots, N-1$, we may analogously to the stochastic case define a linear operator $\cA: \cmeas(\TT) \to \CC^N$ such that
\begin{align} \label{eq:signal_relation}
    \bx = \cA(\detmeas)
\end{align}
where $\bx = \begin{bmatrix} x_0 & x_1 & \ldots & x_{N-1} \end{bmatrix}^T$. It may here be noted that the problem of inverting the relation \eqref{eq:signal_relation} under the assumption of $\mu$ consisting of a small set of point-masses corresponds to the classical works on super resolution spectral estimation \cite{candes2014towards} via atomic norm/total variation minimization \cite{condat2020atomic}. If $\TT$ is discretized, this simply corresponds to $\ell_1$ minimization. %\fe{[Cite also Condat's "Atomic norm minimization..." somewhere here.]}.
For our purposes, we are interested in connecting the signal $x$ to a mass distribution. It may be noted that although $\detmeas \in \cmeas(\TT)$, it is not an element of $\pmeas(\TT)$. However, we may construct such an element by the variation $\abs{\detmeas}~\in~\pmeas(\TT)$. Then, if the supports of the Dirac measures in \eqref{eq:deterministic_measure} are completely disjoint, $\abs{\detmeas}$ is simply given by
\begin{align} \label{eq:simple_tv_measure}
    \abs{\detmeas} =  \sum_{k=1}^K\sum_{\ell = 1}^{L^{(k)}} \abs{\alpha_\ell^{(k)}}\delta_{\omega_\ell^{(k)}}.
\end{align}
If this is the case, then $\abs{\detmeas}$ can be decomposed into $K$ component measures $|\detmeas^{(k)}| \in \pmeas(\TT)$ as
\begin{align} \label{eq:det_decomposition}
    \abs{\detmeas} = \sum_{k=1}^K |\detmeas^{(k)}|
\end{align}
where, as the notation suggests, $|\detmeas^{(k)}|$ is indeed the variation of $\detmeas^{(k)}$. However, if the support of some Dirac measures would overlap, the representation in \eqref{eq:simple_tv_measure} is no longer valid, nor is the decomposition into the $K$ components. Instead, we get the inequality
\begin{align}\label{eq:tv_inequality}
    \abs{\detmeas}(\mathcal{E}) \leq \sum_{k=1}^K |\detmeas^{(k)}|(\mathcal{E})
\end{align}
which holds for any $\mathcal{E} \subset \TT$. Although this is a subtle difference, it will have practical consequences in the design of the algorithms implementing our proposed estimators.

To summarize, for both the stochastic and deterministic signal models, we can connect signal observations to elements of $\cmeas(\TT)$ via relations \eqref{eq:cov_relation} and \eqref{eq:signal_relation}, with complex measures encoding the signal's frequency content. The aim is then to recover the (approximately) harmonic structures of $\stochmeas$ and $\detmeas$. The tool for achieving this will be regularizers constructed based on the concept of optimal transport, described next.
\section{Optimal transport}
Consider two general, non-negative measures, or mass distributions, $\genmeas_1, \genmeas_2 \in \pmeas(X)$ defined on some space $X$, with same total mass, i.e., $\genmeas_1(X) = \genmeas_2(X)$. Then, the Monge-Kantorovich problem of OT considers the most cost-efficient way of transporting, or morphing, $\genmeas_1$ as to be identical to $\genmeas_2$. The cost of moving one unit-mass from a support-point $\omega_1$ of $\genmeas_1$ to a corresponding support-point $\omega_2$ of $\genmeas_2$ is given by $\groundcost(\omega_1,\omega_2)$, where $\groundcost: X\times X \to \RR$ is a so-called ground-cost function. For example, if $X = \RR^d$ for some $d \in \NN$, then $\groundcost$ is typically picked as the squared Euclidean distance. The problem of minimizing the total cost of rearrangement is then stated as
\begin{equation} \label{eq:basic_ot}
\begin{aligned}
    \minwrt[m \in \pmeas(X\times X)]& \left\{ \langle m, c \rangle \triangleq \int_{X\times X} c(\omega_1,\omega_2) dm(\omega_1,\omega_2) \right\}
    \\ \text{s.t. }& \int_{\mathcal{E}_1 \times X} dm(\omega_1,\omega_2) = \int_{\mathcal{E}_1} d\mu_1(\omega),
    \\& \int_{X \times \mathcal{E}_2} dm(\omega_1,\omega_2) = \int_{\mathcal{E}_2} d\mu_2(\omega),
    \\& \text{for all sets }\mathcal{E}_1,\mathcal{E}_2 \subset X.
\end{aligned}
\end{equation}
Here, the transport plan $m$ describes how mass is moved from $\genmeas_1$ to $\genmeas_2$, with the constraints ensuring that $m$ indeed describes transport between the two distributions. This can be phrased as $m$ having margins $\mu_1$ and $\mu_2$, and we use the short-hand $\projop_1(m) = \genmeas_1$ and $\projop_2(m) = \genmeas_2$, where the projection operators $\projop_1$ and $\projop_2$ are uniquely defined by the constraints in \eqref{eq:basic_ot}.

The minimal objective in \eqref{eq:basic_ot} serves as a notion of distance or dissimilarity on $\pmeas(X)$, which has been exploited in problem machine learning and signal processing (see \cite{KolouriPTSR17_34} for a review). Most closely related to the multi-pitch estimation problem considered in this work, OT has been earlier used for solving inverse problems of spectral estimation \cite{elvander2020multi,elvander2020multi}. In contrast to the problem \eqref{eq:basic_ot}, we are in this work interested in the setting where we can only observe one of the marginals indirectly according to \eqref{eq:cov_relation} or \eqref{eq:signal_relation} and where the \textit{second margin is free or unconstrained} in the sense that it is not related to observations or measurements. Instead, we will extend the transport problem \eqref{eq:basic_ot} as to allow us to ultimately recover decompositions on the form \eqref{eq:stoch_decomposition} or \eqref{eq:det_decomposition}.
%
%

%%%%%%
\section{Optimal transport for harmonic spectra}
Consider any $\genmeas \in \pmeas(\TT)$. Building on \cite{Elvander2020HarmonicDef}, we may quantify how close $\genmeas$ is to being the power spectrum of a harmonic signal with a certain fundamental frequency $\omega_0'$ by means of an OT problem. In particular, let $\groundcost: \TT\times\TT \to \RR$ be a ground-cost function, where $\groundcost(\omega,\omega_0)$ is the cost of associating the frequency $\omega$ with a fundamental frequency $\omega_0$. Then, we may formally solve
\begin{equation} \label{eq:single_pitch_transport}
\begin{aligned}
     \minwrt[m \in \pmeas(\TT\times\TT)] \quad \langle m, c\rangle, \text{ s.t. } \projop_1(m) = \genmeas,\; \projop_2(m) = \mu(\TT)\delta_{\omega_0'}.
\end{aligned}
\end{equation}
Note here that the first margin of the transport plan $m$ co-incides with $\genmeas$, whereas the second margin has all its mass located at the single point $\omega_0'$. Then, if $m^\star$ is a solution to \eqref{eq:single_pitch_transport},  $\harmotdist^c(\omega_0',\genmeas) \triangleq \langle m^\star, c\rangle$ serves as a notion of distance between $\genmeas$ and (any) harmonic spectrum with fundamental frequency $\omega_0'$. It may be noted that this minimal objective can be computed in closed form. In particular,
\begin{align*}
    \harmotdist^c(\omega_0',\genmeas) = \langle m^\star, c\rangle = \int_\TT c(\omega,\omega_0')d\genmeas(\omega).
\end{align*}
It should here be stressed that $\projop_2(m^\star)$ is not the spectrum of a harmonic signal. However, such a spectrum can be recovered from the pair $(m^\star, c)$, or from $(\genmeas,c)$, and depends on the specific construction of the ground-cost $c$. In particular, in \cite{Elvander2020HarmonicDef}, the squared distance to the closest harmonic was used, i.e., $\groundcost = \hat{\groundcost}$, where
\begin{align}
    \hat{\groundcost}(\omega,\omega_0) = \min_{k \in \ZZ_+} (\omega- k\omega_0)^2,
\label{eq:gc_old}
\end{align}
where $\ZZ_+$ is the set of positive integers. The implied harmonic spectrum is then given by
\begin{align*}
    \genmeas_{\text{h}} = \sum_{k} \sigma^2_k\delta_{k \omega_0}, \text{ with }\sigma^2_k = \int_{(k-1/2)\omega_0}^{(k+1/2)\omega_0 } d\genmeas(\omega).
\end{align*}
It may here be noted that if $\genmeas$ contains point-masses on the interval edges, then $\genmeas_\text{h}$ is not unique. The same ground-cost was used in \cite{flamary2016optimal} in the context of music transcription. Although the choice $\hat{\groundcost}$ is intuitively reasonable, it leads to a strong preference for low fundamental frequencies $\omega_0$ if one is to use $\harmotdist^{\hat{\groundcost}}$ to choose between different $\omega_0$ fitting $\genmeas$. In particular, it may be readily verified that for any $\genmeas \in \pmeas(\TT)$ and $\omega_0 \in \TT$,
\begin{align*}
    \harmotdist^{\hat{\groundcost}}(\omega_0/2, \genmeas) \leq \harmotdist^{\hat{\groundcost}}(\omega_0, \genmeas),
\end{align*}
which is a direct consequence of $\hat{\groundcost}(\omega_0/2,\omega) \leq \hat{\groundcost}(\omega_0,\omega)$ for all $\omega \in \TT$. This inevitably leads to so-called sub-octave errors and is typically countered using heuristic rules for limiting the maximal harmonic order or imposing additional structure on the implied spectrum $\genmeas_\text{h}$. In this work, we propose to remedy this by a simple normalization and define the ground-cost as
\begin{align}\label{eq:our_groundcost}
    \groundcost(\omega,\omega_0) = \frac{1}{\omega_0^2} \hat{\groundcost}(\omega,\omega_0) = \min_{k \in \ZZ_+} \left( \frac{\omega}{\omega_0}- k\right)^2.
\end{align}
Note that, just as for $\hat{\groundcost}$, $\groundcost$ induces transport plans $m$ where mass is transported to the closest harmonic of a given $\omega_0$. Thus, $\groundcost$ and $\hat{\groundcost}$ imply the same harmonic spectrum $\genmeas_\text{h}$. However, the corresponding $\harmotdist^{{\groundcost}}$ shows preference for $\omega_0$ over $\omega_0/2$ when applied to inharmonic signals. To see this, consider an inharmonic component located at $\omega_k = k\omega_0 + \Delta$ for some $k$, where $\Delta$, with $\abs{\Delta} < \omega_0/2$, is the inharmonic deviation. Then, we have that $\groundcost(\omega_k, \omega_0) <  \groundcost(\omega_k, \omega_0/2)$ if and only if $\abs{\Delta} < \omega_0/3$. That is, in order to prefer $\omega_0/2$ over $\omega_0$, the inharmonic deviations have to be substantial. An illustration of this is presented in Figure~\ref{fig:cost_example_2}, where both functions are plotted for the pitch candidates $\omega_0$ and $\omega_0/2$. 
Thus, for the proposed ground-cost function, the optimal solution with small harmonic deviations is obtained by the model with the \textit{highest possible fundamental frequency} that can, in close proximity, describe all harmonics.
%
%%%%%
\begin{figure}[t]
    \centering
    \includegraphics[width = \linewidth]{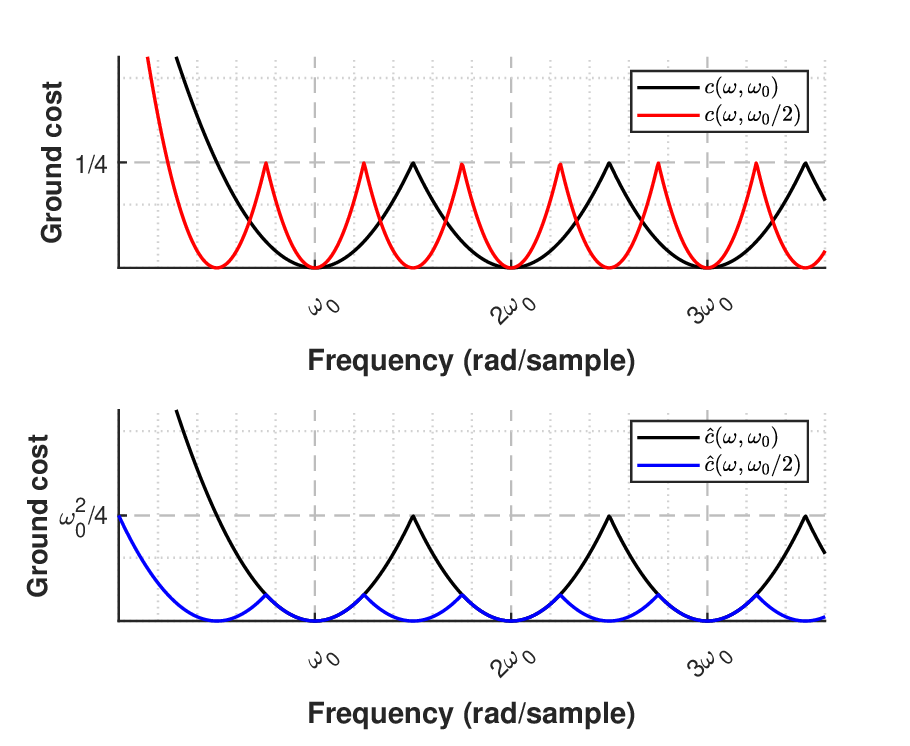}
    \caption{ Example plot of the ground-cost function in \eqref{eq:gc_old} in the lower plot, and example plot of the proposed ground-cost function in \eqref{eq:our_groundcost} in the upper plot. }  
    \label{fig:cost_example_2}\vspace{-1mm}
\end{figure}
%%%%%
%%%%%

\section{Harmonic clustering}
In this section, we derive the proposed multi-pitch estimators, one for the stochastic and one for the deterministic case. For clarity of exposition, we focus the discussion on the simpler stochastic setting and then show the necessary modifications for the deterministic one.
\subsection{Ideal problem and convex relaxation}
For the model in \eqref{eq:almost_harmonic_signal}, assume that the number $K$ of pitch components is known. Letting $\bomega_0 = \{ \omega_0^{(k)} \}_{k=1}^K$ be the set of (unknown) fundamental frequencies, we would like to solve
\begin{align}\label{eq:ideal_clustering_and_estimation}
    \minwrt[\stochmeas^{(k)} \in \pmeas(\TT),\; \bomega_0]\quad \frac{1}{2} \norm{ \br - \cA(\stochmeas))}_2^2 + \sum_{k=1}^K \harmotdist^{\groundcost}(\omega_0^{(k)}, \stochmeas^{(k)}),
\end{align}
where $\stochmeas = \sum_{k=1}^K \stochmeas^{(k)}$ and where the first data-fit term takes into account that $\br$ will in practice be a noisy estimate of the finite-length covariance sequence. Note that each term of the regularizing function quantifies how well a component measure, or spectrum, $\stochmeas^{(k)}$ fits a harmonic spectrum with fundamental frequency $\omega_0^{(k)}$. It may be verified that although this problem is convex in the set of $\stochmeas^{(k)}$, it is not jointly convex in $(\{\stochmeas^{(k)} \}_k,\bomega_0)$. In fact, it is not even marginally convex in $\bomega_0$ due to the structure of the ground-cost $\groundcost$ in \eqref{eq:our_groundcost}. Although, one could envision addressing this problem with the iterative tools from \cite{elvander2025mixtures}, this would require a good initial point for $\bomega_0$. To work around this, consider the case when $\bomega_0$ is fixed and \eqref{eq:ideal_clustering_and_estimation} is minimized with respect to the $\stochmeas^{(k)}$. This can equivalently be formulated as
\begin{equation} \label{eq:non_sparse_stoch_ot}
\begin{aligned}
    \minwrt[m \in \pmeas(\TT \times \Omega)]\quad \frac{1}{2} \norm{ \br - \cA(\projop_1(m))}_2^2 + \langle m,c\rangle
\end{aligned}
\end{equation}
where the transport plan $m \in \pmeas(\TT \times \Omega)$ describes transport between $\TT$ and the discrete grid $\Omega = \{ \omega_0^{(1)}, \ldots, \omega_0^{(K)}\}$. This is a direct extension of the single-pitch problem in \eqref{eq:single_pitch_transport}; if $\Omega$ is a singleton set then the problems are the same. It may be noted that the individual $\stochmeas^{(k)}$ can be retrieved from $m$ as
\begin{align*}
    %\stochmeas^{(k)} = m^{(k)}
    \stochmeas^{(k)} = \projop^{(g)}(m)
\end{align*}
where the projection $\projop^{(g)}: \pmeas(\TT\times\Omega) \to \pmeas(\TT)$ is defined by the relation
\begin{align*}
    \projop^{(g)}(m)(\mathcal{E}) = m(\mathcal{E} \times \{\omega_0^{(k)}\}), \; \forall \mathcal{E} \subset \TT.
\end{align*}
Intuitively, the projected, or \emph{restricted}, measure $\projop^{(g)}(m)$ can be thought of as fixing the second argument of the transport plan $m$ to $\omega_0^{(k)}$.

With these constructions, we may extend $\Omega$ to a large grid $\Omega_G$ consisting of $G \gg K$ fundamental frequencies. We would want the support over $\Omega_G$ to be sparse, i.e., $m^{(g)}(\TT) = 0$ for most $g \in \indset{G}$, an illustrative example of which can be seen in Figure \ref{fig:transplan}. As to promote transport plans $m$ that are sparse in this sense, while avoiding the combinatorial problem arising for counting the number of non-zero $\projop^{(g)}(m)$, we propose to use the convex proxy

\begin{align*}
    \norm{m}_{\infty,1} = \sum_{g=1}^G \norm{ \projop^{(g)}(m)}_\infty,
\end{align*}
where we with some slight abuse of notation over-load the mixed-norm notation often used for matrices. For elements $\stochmeas \in \pmeas$ containing only point-masses, we by $\norm{\stochmeas}_{\infty}$ mean
\begin{align*}
    \norm{\stochmeas}_{\infty} = \max_{\omega\in \TT} \:\stochmeas(\omega),
\end{align*}
corresponding to the coefficient of the largest point-mass. It may be noted that this can be directly extended to the case when $\stochmeas$ also contains a component with a density.

With this, we are ready to state our proposed multi-pitch estimator for the stochastic setting. We define it as the solution to the optimization problem
\begin{equation} \label{eq:stoch_estimator}
\begin{aligned}
    \minwrt[m \in \pmeas(\TT \times \Omega_G)]\: \frac{1}{2} \norm{ \br \!-\! \cA(\projop_1(m))}_2^2 + \zeta\!\left(\langle m,c\rangle + \eta\norm{m}_{\infty,1} \right),
\end{aligned}
\end{equation}
where $\zeta$ and $\eta$ are positive regularization parameters controlling the trade-off between data-fit and regularization, and sparsity of the transport plan. Note here that increasing $\eta$ has the effect of promoting transport plans $m$ concentrated to a small number of points of $\Omega_G$. Assume that an optimal $m^\star$ have been obtained from solving \eqref{eq:stoch_estimator}. Then, by inspecting the restricted measures $\projop^{(g)}(m)$ we can identify the pitches present in the signal, and their corresponding fundamental frequencies $\omega_0^{(g)}$. Furthermore, we can then restrict the grid $\Omega_G$ to contain only these active pitches and solve \eqref{eq:non_sparse_stoch_ot} as to obtain a debiased solution, countering the shrinkage towards zero caused by the $\norm{\cdot}_{\infty,1}$-term. When designing an algorithm for solving \eqref{eq:stoch_estimator}, we will work with the equivalent formulation
\begin{equation} \label{eq:stoch_estimator_hidden}
\begin{aligned}
    \minwrt[\stochmeas \in \pmeas(\TT)]\quad \frac{1}{2} \norm{ \br - \cA(\stochmeas)}_2^2 + \zeta \stochreg(\stochmeas),
\end{aligned}
\end{equation}
where $\stochreg(\stochmeas)$ is defined as
\begin{equation} \label{eq:stoch_regularizer}
\begin{aligned}
\stochreg(\stochmeas) = \mminwrt[m \in \pmeas(\TT \times \Omega_G)] & \quad \langle m,c\rangle + \eta \norm{m}_{\infty,1}
\\ \text{s.t. } &\quad \projop_1(m) = \stochmeas.
\end{aligned}
\end{equation}
When solving \eqref{eq:stoch_estimator_hidden} in practice, we will discretize the operator $\cA$. This highlights a key difference between the method proposed herein and other multi-pitch estimators: the measurement operator is independent of the set of pitch candidates, corresponding to the set $\Omega_G$. For methods such as \cite{adalbjornsson2013PEBS, ADALBJORNSSON2015236, ELVANDER201656}, the operator $\cA$ is constructed based on the set of candidate pitches, causing coherence problems as mentioned earlier.
\begin{figure}[t]
    \centering
    \includegraphics[width = \linewidth]{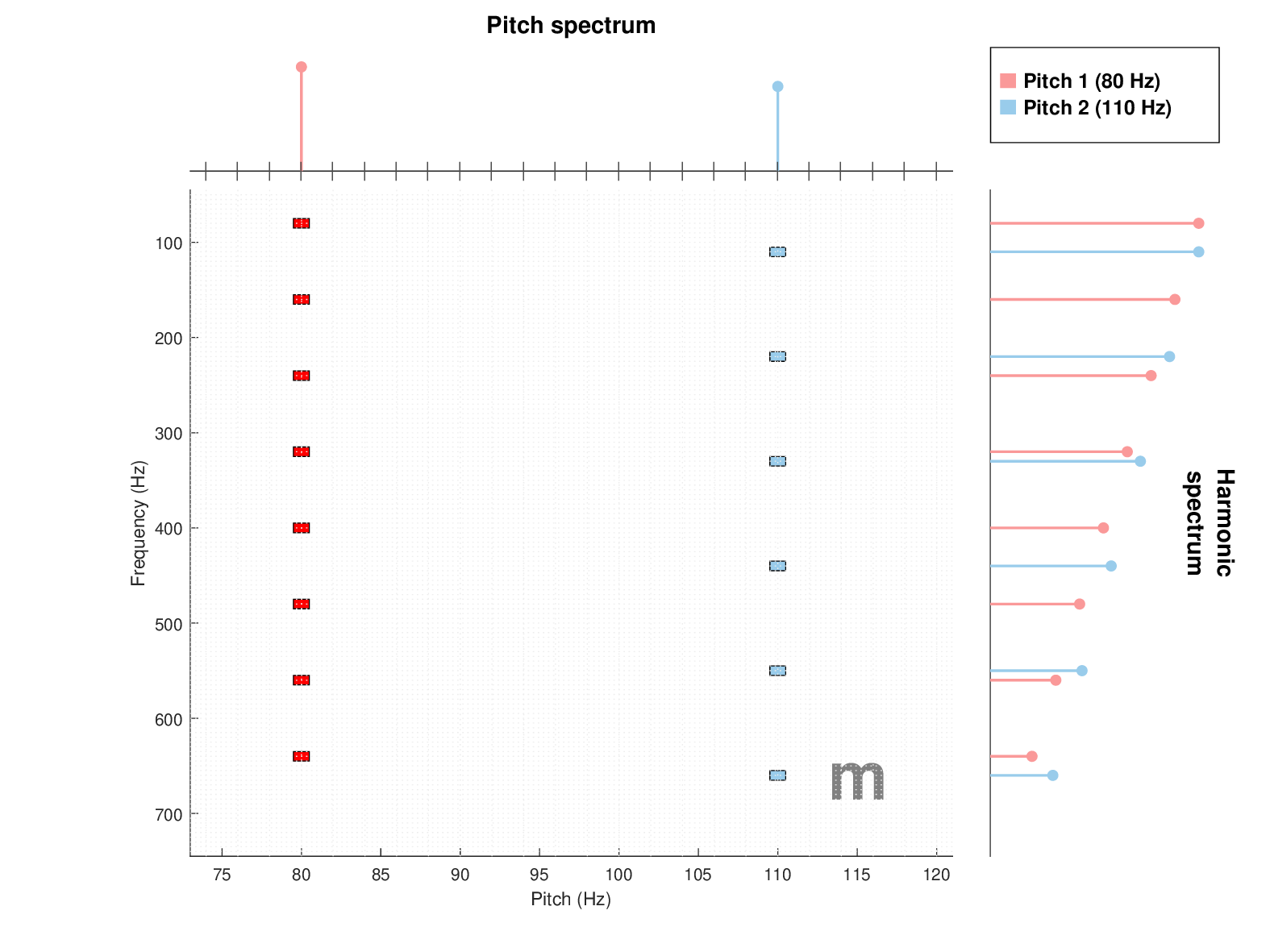}
    \caption{ Illustrative example of a transport plan between a harmonic spectrum to a pitch spectrum. The upper plot is the pitch spectrum, while the rotated plot on the right-hand side displays the spectrum of the evaluated signal.}  
    \label{fig:transplan}\vspace{-1mm}
\end{figure}
\subsection{Modification for deterministic model}
Recall that for the deterministic model, the object of interest is $\detmeas \in \cmeas(\TT)$. In order to arrive at an estimator analog to the stochastic one, we have to extend the formulation \eqref{eq:stoch_estimator_hidden} rather than \eqref{eq:stoch_estimator}. The reason for this is that the projection of the transport plan $m$ in \eqref{eq:stoch_estimator} necessarily is an element of $\pmeas(\TT)$. However, the extending formulation \eqref{eq:stoch_estimator_hidden} allows us to use the variation construction introduced earlier. Recall that if $\detmeas \in \cmeas(\TT)$ can be decomposed into as set of components $\detmeas^{(k)}$, then their corresponding variations $\abs{\detmeas} \in \pmeas(\TT)$ and $|\detmeas^{(k)}| \in \pmeas(\TT)$ satisfy the inequality \eqref{eq:tv_inequality}, where equality holds if the supports of the components $\detmeas^{(k)}$ do not overlap, or have the same phase on any overlaping subset of $\TT$. Now, consider a transport plan $m \in \pmeas(\TT \times \Omega_G)$. In analog to the interpretation of the restricted measure $\projop^{(g)}(m)$ as representing a component measure $\stochmeas^{(g)}$ in the stochastic case, we can in the deterministic case view $\projop^{(g)}(m)$ as representing the \emph{variation} $|\detmeas^{(k)}|$ of a component $\detmeas^{(k)} \in \cmeas(\TT)$. In view of this and the inequality \eqref{eq:tv_inequality}, it becomes natural to replace the equality constraint in \eqref{eq:stoch_regularizer} with the inequality
\begin{align} \label{eq:projection_inequality}
    \projop_1(m) \geq \abs{\detmeas},
\end{align}
which should be understood as a point-by-point inequality on $\TT$. We can then define the regularizer $\detreg: \cmeas(\TT) \to \RR$ as
\begin{equation} \label{eq:det_regularizer}
\begin{aligned}
\detreg(\detmeas) = \mminwrt[m \in \pmeas(\TT \times \Omega_G)] & \quad \langle m,c\rangle + \eta \norm{m}_{\infty,1}
\\ \text{s.t. } &\quad \projop_1(m) \geq \abs{\detmeas}.
\end{aligned}
\end{equation}
It may be verified that $\detreg$ is a convex function on $\cmeas(\TT)$. Note also that convexity is lost if the inequality is replaced by equality in the constraint. With this, we can formulate the multi-pitch estimator for the deterministic model as the solution to
\begin{equation} \label{eq:det_estimator_hidden}
\begin{aligned}
    \minwrt[\detmeas \in \cmeas(\TT)]\quad \frac{1}{2} \norm{ \by - \cA(\detmeas)}_2^2 + \zeta \detreg(\detmeas).
\end{aligned}
\end{equation}
As in the stochastic case, this is a convex optimization problem. Furthermore, we may retrieve the frequency content of the component pitches by inspecting the restricted measures $\projop^{(g)}(m)$. It should however be noted that we in general can only hope to retrieve $|\detmeas^{(k)}|$ rather than $\detmeas^{(k)}$ as we lose phase-information at points in $\TT$ where the inequality in \eqref{eq:projection_inequality} is strict. However, at points where equality holds, which for example is the case when the support of the components do not overlap, the phase can also be retrieved.
\section{Discretization and entropy regularization}
Expanding upon the problems in \eqref{eq:stoch_estimator_hidden} and \eqref{eq:det_estimator_hidden}, we here aim to discretize both problems, allowing for efficient implementations. To this end, we introduce the grid $\bomega =\{ \omega_1,\ldots, \omega_F\}\subset [0,\pi)$ consisting of $F$ frequencies, corresponding to a discretization of the frequency space, $\TT$. Then, the problem in \eqref{eq:stoch_estimator_hidden} can be written in vector form as 
\begin{equation}
\begin{aligned}
    \minwrt[\bnu \in \RR^F_+] & \quad \frac{1}{\Tlag} \norm{\bhr - \bA\bnu}_2^2 + \beta (\bnu^T \onevec_F)+ \zeta \bS_c(\bnu),
    \label{eq:main_problem_spectrum}
\end{aligned}
\end{equation}
where $\bA = \begin{bmatrix} \ba(\omega_1) & \dots & \ba(\omega_F)\end{bmatrix}$ is a dictionary matrix which is a matrix whose columns $\ba(\omega_f) \in \mathbb{C}^N$ consists of Fourier vectors corresponding to frequencies $\omega_f$, $f \in \indset{F}$, over the sample time. The matrix multiplication $\bA\bnu$ here corresponds to the operation $\cA(\stochmeas)$ in \eqref{eq:stoch_estimator_hidden}, where $\bnu \in \RR_+^{F}$ is a vectorized version of $\stochmeas$.  

Additionally, $\bhr \in \CC^{\Tlag}$ denotes the (partial) sample-auto-covariance vector, formed based on noisy measurements.\footnote{In practice when solving \eqref{eq:main_problem_spectrum}, we consider the equivalent purely-real problem where we stack the real- and imaginary parts according to $\hat{\br}_R = \begin{bmatrix}
    \mathfrak{Real}(\hat{\br}) & \mathfrak{Imag}(\hat{\br})
\end{bmatrix}^T$ and correspondingly for $\bA$.} To combat the addition of noise, we propose an $\ell_1$-norm regularizer on $\bnu$. The $\ell_1$-term promotes sparsity in the frequency estimates, pushing the components driven only by noise down to $0$. Since $\bnu$ contains only non-zero elements, the resulting $\ell_1$-norm term becomes $\beta (\bnu^T \onevec_F)$. We also introduce a normalizing factor of $2/\Tlag$ for the first term. Lastly, the last regularizing term is defined as
\begin{align*}
    \bS_c(\bnu) = \min_{\bM \in \RR_+^{F\times G}  } &\langle \bC , \bM \rangle + \epsilon D(\bM)+ \eta \norm{\bM}_{\infty,1},
    \\ \text{s.t. }& \bM\onevec_G = \bnu,
\end{align*}
corresponding to a discretization of $\stochreg(\stochmeas)$, where $\bM$ is the discretized transport plan, $\bC \in \RR_+^{F\times G}$ is the cost \textit{matrix}, with $\bC_{f,g} = c(\omega_f,\omega_0^{(g)})$, and where $\norm{\cdot}_{\infty,1}$ now denotes the regular mixed norm, i.e., $\norm{\bM}_{\infty,1} = \sum_{g=1}^G \max_{f \in \indset{F}} \abs{[\bM]_{f,g}}$. Furthermore, $\epsilon D(\bM) = \epsilon\sum_{f,g} [\bM]_{f,g}\log[\bM]_{f,g} - [\bM]_{f,g} + 1$, with $\epsilon > 0$, is an entropic regularization term \cite{cuturi2013sinkhorn} which will allow us to derive efficient numerical solvers. It may be noted that the addition of $\epsilon D(\bM)$ renders \eqref{eq:main_problem_spectrum} strictly convex.

Similarly, we discretize the problem in \eqref{eq:det_estimator_hidden}, yielding
\begin{align}
    \label{eq:main_problem}
    \minwrt[\bmu \in \CC^F] \quad \frac{1}{N} \norm{\by - \bA \bmu}_2^2 + \beta \norm{\bmu}_1 + \zeta \otdist(\bmu),
\end{align}
where
\begin{equation}
\label{eq:otdist}
    \begin{aligned}
        \otdist(\bmu) \triangleq \mminwrt[\bM \in \RR_+^{F\times G}]&\;\; \langle \bC, \bM \rangle + \epsilon D(\bM) +  \eta \norm{\bM}_{\infty,1},
        \\\text{s.t. }&\quad \bM \onevec_G \geq \abs{\bmu}.
    \end{aligned}
\end{equation}
The matrix multiplication $\bA\bmu$ here corresponds to the operation $\cA(\detmeas)$ in \eqref{eq:det_estimator_hidden}, where $\bmu \in \RR_+^{F}$ is a vectorized version of $\detmeas$. As we again work with noisy measurements, the $\ell_1$-norm term is used here as well. However, as the elements of $\detmeas$ are complex valued, the $\ell_1$-norm term is not simplified, as was the case for the stochastic problem. In the following section, we introduce efficient solvers for the estimator presented herein.
\section{Numerical algorithms}
In order to compute the solutions to \eqref{eq:main_problem_spectrum} and \eqref{eq:main_problem}, we use iterative algorithms. In particular, we in this work propose to use proximal gradient-type methods. As a brief general outline, consider the minimization problem
\begin{align*}
    % \label{eq:min_probl}
    \minwrt[x \in \CC^F] \quad E(x) = e(x) + d(x),
\end{align*}
where $e$ is a smooth convex function, and $d$ is a (possibly non-differentiable) convex function. Under the assumption that $e$ has a Lipschitz continuous gradient $L$, the problem can be solved through alternating between taking gradient steps with respect to $e$ and computing the proximal mapping of the regularizer $d$, given by
\begin{align}
\label{eq:prox_grad_scheme}
     x^{(j+1)} &= \mathrm{prox}_{\gamma d}(x^{(j)} - \gamma\nabla e(x^{(j)}))\\
     &\triangleq \argminwrt[x] \quad \frac{1}{2} \| x - u \|_2^2 + \gamma d(x). \nonumber
\end{align}
with $1/\gamma \geq L$ being the step size, $u =  (x^{(j)} - \gamma\nabla e(x^{(j)}) )$, and $\mathrm{prox_d}$ is the proximal operator of $d$. A generalization of proximal gradient descent has been proposed (see, e.g., \cite{lu2018relatively, bauschke2017descent, hanzely2021accelerated}), where an approximation of the function $e$ depends on a Bregman distance associated with a function $\varphi$. The iterations for this are then given by
\begin{align}
\label{eq:bregprox}
    x^{(j+1)} = \argminwrt[x\in\CC^F] \left\{ \gamma\langle \nabla e(x^{(j)}),x \rangle + \gamma d(x) + \mathcal{D}_\varphi(x,x^{(j)}) \right\},
\end{align}
where
\begin{align*}
    \mathcal{D}_\varphi(x,x^{(j)}) = \varphi(x) - \varphi(x^{(j)}) - \nabla \varphi (x^{(j)})(x-x^{(j)}).
\end{align*}
It should be noted that with $\varphi(x) = \norm{x}_2^2$, the regular proximal gradient scheme is obtained. Furthermore, both descent methods have so-called \textit{accelerated} versions, which speed up the convergence; see \cite{beck2009fast, hanzely2021accelerated} for more information. Below, we first present the solver for the stochastic model, based on the Bregman proximal gradient schemes presented in this section.

\subsection{Stochastic model}
 We herein propose to solve \eqref{eq:main_problem_spectrum} using the Bregman proximal gradient scheme in \eqref{eq:bregprox}, with $e(\bnu) = \frac{1}{\Tlag} \norm{\bhr - \bA \bnu}_2^2 + \beta (\bnu^T \onevec_F)$, and $d(\bnu) = \zeta\otdistsp(\bnu)$. The function associated with the Bregman divergence we here use is $\varphi(\bnu) = \sum_{f=1}^F \bnu_f \log \bnu_f$, yielding the Kullback-Leibler relative entropy, given by
    \begin{align*}
        \mathcal{D}_{KL}(\bnu, \bnu^{(j)}) = \sum_{f=1}^F [\bnu]_f \log( [\bnu]_f/[\bnu^{(j)}]_f) + [\bnu^{(j)}]_f - [\bnu]_f.
    \end{align*}
    We here choose the step size to be $\stepsize = \frac{1}{L}$, where $L~=~2\norm{\bA}^2/\Tlag$, and where $\norm{\cdot}$ is the operator norm, as is also the case for the other method presented later. For this method, this does not guarantee convergence, but was shown by empirical testing to be a sufficiently small step size. Inspired by the work in \cite{KarlssonRingh17_gen}, the iterates for the solution are, nonetheless, given by
\begin{align} \label{eq:bregman_iteraiton}
    \bnu^{(j+1)} = \argminwrt[\bnu \in \RR_+] \quad &\stepsize\left\langle\bu ,\bnu \right\rangle + \mathcal{D}_{KL}(\bnu, \bnu^{(j)})+ \stepsize\zeta \otdistsp(\bnu),
\end{align}
where $\bu = \nabla\left(\frac{1}{\Tlag}\| \bhr - \bA\bnu^{(j)} \|_2^2 + \stepsize\beta ((\bnu^{(j)})^T \onevec_F\right)$. For this minimization problem, the following proposition holds.
%%%%% Prop: bregman prox op %%%%%%
\begin{proposition}
\label{prop:prox_op_bregman}
Let $\odot$ denote elementwise multiplication. Then, the update $\bnu^{(j+1)}$ in \eqref{eq:bregman_iteraiton} and corresponding transport plan $\bM$ are given by
    \begin{align} 
        \bnu &= \bnu^{(j)}\odot \exp\left(-\stepsize\bu - \stepsize\zeta\blambda\right),\label{eq:optimal_nu}
        \\\bM &= \mathrm{diag}(\bv) ( \bK \odot \bW),\label{eq:optimal_stoch_M}
    \end{align}
with $\bv~\!=\!~\exp\left( \frac{1}{ \epsilon} \blambda \right)$, $\bK = \exp\left( -\frac{1}{\epsilon} \bC \right)$, $\bW = \exp\left( \frac{1}{\epsilon} \bPsi \right)$, and where $\blambda \in \RR^F_+$ and $\bPsi \in \RR^{F\times G}$ solve
    \begin{equation}\label{eq:prox_dual_problem_cov}
    \begin{aligned}
    \minwrt[\substack{\blambda,\; \bPsi \\ \|\bPsi\|_{1,\infty} \leq \eta}] \;\stepsize\zeta \epsilon \langle\bK\odot\bW, \bv\onevec_G^T \rangle
        +\langle \bnu^{(j)}, \exp\left(-\stepsize\bu - \stepsize\zeta\blambda\right)\rangle
    \end{aligned}
 \end{equation}
    Here, $\|\bPsi\|_{1,\infty} = \max_g \sum_f \abs{\Psi_{f,g}}$ is the dual norm of $\|\cdot\|_{\infty,1}$, and all exponentiation and powers are elementwise.
    \begin{proof}
        See appendix.
    \end{proof}
\end{proposition}
%%%%%%
For solving the problem in \eqref{eq:prox_dual_problem_cov}, we use a block-coordinate descent scheme (see \cite{KarlssonRingh17_gen,elvander2020multi}). The following proposition holds.
\begin{proposition}\label{prop:stochastic_updates}
Consider an iterative scheme with update in step $k+1$ given as
\begin{align*}
    \blambda^{k+1} &=\frac{\epsilon}{1 + \stepsize\zeta\epsilon}\left( \log(\bnu_0) - \log(\bxi^k) \right),
    %%%
    \\
    \bPsi^{k+1} &= \argminwrt[\bPsi: \| \bPsi \|_{1,\infty} \leq \eta]  \left\langle (\bv^{k+1} \onevec_G^T) \odot \bK, \exp\left(\frac{1}{\epsilon} \bPsi\right) \right\rangle,%\\
   % \quad & \text{s.t. } \quad \| \bPsi \|_{1,\infty} \leq \theta\eta,
\end{align*}
where $\bnu_0 = \bnu^{(j)} \odot \exp(-\stepsize \bu)$, $\bxi^k = \exp(-\frac{1}{\epsilon} \bC + \frac{1}{\epsilon} \bPsi^k) \onevec_G$, and $v^{k+1} = \exp(\frac{1}{\epsilon} \blambda^{k+1})$
Then, the iterates converge linearly to a solution of \eqref{eq:prox_dual_problem_cov}.
\end{proposition}
\begin{proof}
    The objective function of \eqref{eq:prox_dual_problem_cov} is differentiable with respect to the dual variable $\blambda$. For a fixed $\bPsi$, the stated update for $\blambda$ satisfies the zero-gradient equations and is this optimal. Then, for fixed $\blambda$, we may minimize the objective with respect to $\bPsi$ using \cite[Thm. 2]{Haasler2024}, which only requires a sorting operation.
As the problem satisfies the assumptions of \cite[Thm. 2.1]{luo1992convergence}, linear convergence follows directly.
\end{proof}
It may be noted that we may retrieve the distribution over $\Omega$, i.e., the set of candidate fundamental frequencies, as
\begin{align*}
    \bM^T \onevec_F = (\bK^T \odot \bW^T) \bv.
\end{align*}
From inspecting the support of this vector, one can identify which pitches that are present in the signal. Furthermore, we may characterize the solution as the entropy regularization parameter $\epsilon \to 0$.
\begin{proposition} \label{prop:maximum_entropy}
Let $\{\epsilon_p\}_p$ be a sequence with $\lim_{p\to \infty} \epsilon_p \to 0$. Then, the corresponding sequence $\{\bnu_p\}_p$ of solutions to \eqref{eq:main_problem_spectrum} has a limit point $\bnu^\star$ satisfying $\bnu^\star = \bM^\star\onevec_G$, where $\bM^\star$ is the unique maximum-entropy element among all transport plans solving \eqref{eq:main_problem_spectrum} without entropy regularization.
\end{proposition}
\begin{proof}
See appendix
\end{proof}
This result is a direct analog of the classic fact for the case of transport plans where the margins are observable. Furthermore, if we let $\eta = 0$, i.e., not enforce any group-sparsity among the pitches, then it follows directly that each row of $\bM^\star$ consists of only zeros and equal-magnitude elements. Specifically, an element $[\bM^\star]_{k,\ell}$ is non-zero only if $[\bC]_{k,\ell} = \min_{j} [\bC]_{k,j}$. For $k$ with $[\bnu^\star]_k > 0$, the mass is split equally among the elements $[\bM]_{k,\ell}$ with $\{\ell \mid [\bC]_{k,\ell} = \min_{j} [\bC]_{k,j} \}$. Thus, when $\eta = 0$, the power corresponding to a certain frequency is distributed uniformly over the set of pitches most consistent with it. Next, we present the estimator for the deterministic model, where we instead use the regular proximal gradient scheme.
\subsection{Deterministic model}
We propose solving \eqref{eq:main_problem} through a proximal gradient scheme, alternating between taking gradient steps for the data-fit terms, and solving the proximal operator for the remaining terms. Here, in order to guarantee convergence, the stepsize is chosen to be $\gamma = \frac{1}{L}$, where $L = 2 \norm{\bA}^2/N$. Then, the iterates are given by
\begin{align*}
    \bmu^{(j+1)} 
    \!\!=\! \mathrm{prox}_{\stepsize\beta \norm{\cdot}_1 \!+\! \stepsize\zeta\otdist_1(\cdot)}\!\! \left(\!\! \bmu^{(j)} \!-\! \stepsize\nabla_{\bmu^*} \!\frac{1}{N} \norm{\by \!-\! \bA \bmu^{(j)}\!}_2^2\right)\nonumber
\end{align*}
where $\nabla_{\balpha^*}(\cdot)$ denotes the (Wirtinger) gradient with respect to the complex conjugate $\bmu^*$. For the proximal operator, the following proposition holds.
\begin{proposition}
\label{prop:prox_op_abs}
    The proximal operator for $\stepsize\beta \norm{\cdot}_1 + \stepsize\zeta\otdist(\cdot)$  is unique and given by 
    \begin{align*}
        \mathrm{prox}_{\stepsize\beta \norm{\cdot}_1 + \stepsize\zeta\otdist(\cdot)}(\bu) = e^{i\angle \bu}\odot\left( \left(\abs{\bu} - \stepsize\beta\onevec_F - \blambda\right )_+\right)
    \end{align*}
    where $\odot$ denotes elementwise multiplication, $\angle \bu$ denotes the phase angle of $\bu$, $(\cdot)_+ = \max(\,\cdot\,,0)$, and $\blambda~\in~\RR_+^{F}$ solves
    \begin{align} 
    \label{eq:prox_dual_problem_abs}
    \minwrt[\substack{\blambda \in \RR_+^{F} \:,\:  \bPsi \in \RR^{F \times G } \\ \|\bPsi\|_{1,\infty} \leq \eta}] \stepsize\zeta \epsilon \langle\bK\odot\bW, \bv\onevec_G^T \rangle +\frac{1}{2}\!\left( \left(\abs{\bu} \!-\! \stepsize\beta\onevec_F \!-\! \blambda\right)_+\right)^2
\end{align}
where $\bv$, $\bK$, and $\bW$ are as in Proposition~\ref{prop:prox_op_bregman}.
\end{proposition}
The result can be proved in the same way as Proposition~\ref{prop:prox_op_bregman}. Furthermore, \eqref{eq:prox_dual_problem_abs} can be solved using a block-coordinate descent scheme that by \cite[Thm. 2.1]{luo1992convergence} enjoys linear convergence. The explicit iteration updates are summarized in the following proposition.

\begin{proposition}
Consider the iterates
\begin{align*}
    \blambda^{(k+1)} &= \frac{1}{\stepsize\zeta}\abs{\bu} -\frac{\beta}{\zeta}\onevec_F - \epsilon\wrightOmega\left(\bxi^{(k)}\right)\\
    \bPsi^{k+1} &= \argminwrt[\bPsi: \| \bPsi \|_{1,\infty} \leq \eta]  \left\langle (\bv^{k+1} \onevec_G^T) \odot \bK, \exp\left(\frac{1}{\epsilon} \bPsi\right) \right\rangle.%\\
   % \quad & \text{s.t. } \quad \| \bPsi \|_{1,\infty} \leq \theta\eta,
\end{align*}
where
\begin{align*}
    \bxi^{(k)} &= \log\left((\bK \odot \bW^{(k)})\onevec_G\right) + \frac{\abs{\bu} - \stepsize\beta\onevec_F}{\stepsize\zeta\epsilon} - \log\left( \stepsize\zeta\epsilon \onevec_F\right),
    %\label{eq:omega_input_abs}
\end{align*}
and $\wrightOmega(\cdot)$ denotes the (elementwise) Wright omega function \cite{Corless2002}. Then, the iterates converge linearly to a solution of \eqref{eq:prox_dual_problem_abs}.
\end{proposition}
The expressions are obtained by straight-forward algebraic manipulations of the first-order optimality conditions of \eqref{eq:prox_dual_problem_abs} (see \cite{KarlssonRingh17_gen} for a treatment on the Wright-omega function in proximal operators).
\begin{remark}
A result analogous to Proposition~\ref{prop:maximum_entropy} holds also for the deterministic model, i.e., the convergence, as $\epsilon \to 0$, to the maximum entropy solution of the unregularized problem, with the proof just requiring a minor modification.   
\end{remark}
\subsection{Debiasing of $\bnu$ and $\bmu$}
\label{sect:dist_OT}
Once the problems in \eqref{eq:stoch_estimator_hidden} or \eqref{eq:det_estimator_hidden} have been solved, we can identify the pitch components by inspecting the support of $(\bM^\star)^T\onevec_F$, where $\bM^\star$ is the optimal transport plan. However, it should be noted that the corresponding $\bnu$ and $\bmu$ will be biased towards zero due to the sparsity-inducing mixed-norm term. In order to alleviate some of this bias, we can restrict the grid $\Omega_G$ to contain only the identified pitches. We may then re-solve \eqref{eq:stoch_estimator_hidden} and \eqref{eq:det_estimator_hidden} with $\eta = 0$. In fact, it may be readily verified that the expressions in Proposition~\ref{prop:stochastic_updates} hold for this case also, where we note that the iterations converge in one step as $\bPsi \equiv 0$. Then, letting $\epsilon \to 0$, it follows directly from Proposition~\ref{prop:maximum_entropy} that the solution of \eqref{eq:main_problem_spectrum} converges to the maximum-entropy solution of
\begin{align}
    \label{eq:re_est_spec_final}
    \minwrt[\bnu \in \RR^F_+] \quad &\frac{1}{\Tlag} \norm{\bhr - \bA \bnu}_2^2 + \beta \bnu^T\onevec_F+\zeta \bc_{\mathrm{min}}^T\bnu,
\end{align}
where $\bc_{\mathrm{min}} \in \RR^F$ is the vector constructed as $[\bc_{\mathrm{min}}]_k = \min_{\ell} [\bC]_{k,\ell}$. This is a quadratic program with non-negativity constraint and may be solved either using Proposition~\ref{prop:prox_op_bregman} or off-the-shelf solvers. Furthermore, it may be noted that \eqref{eq:re_est_spec_final} is a discretized version of \eqref{eq:non_sparse_stoch_ot} (with fixed $\bomega_0$) with the additional $\beta$-term for suppressing noise. Similarly, for the deterministic model in \eqref{eq:main_problem}, one may use Propostion~\ref{prop:prox_op_abs} as $\epsilon \to 0$ or consider the limiting problem
\begin{align}
    \label{eq:re_est_dir_final}
    \minwrt[\bmu \in \CC^F] \quad &\frac{1}{N} \norm{\by - \bA \bmu}_2^2 + \beta \norm{\bmu}_1+\zeta \bc_{\mathrm{min}}^T\abs{\bmu}.
\end{align}
Note here that the last term implies that any overlapping point-masses are phase-aligned. Furthermore, it may be noted that this can be understood as a LASSO problem with a non-uniformly weighted $\ell_1$-norm. As both estimators have now been obtained, in the following section we conduct experiments on both simulated and real data.
\section{Numerical results}
We first evaluate the proposed methods on simulated data, followed by real data; the next section presents the results for the simulated data, and the following covers the real data.
\subsection{Simulated data}
\begin{figure}[t]
    \centering
    \includegraphics[width = \linewidth]{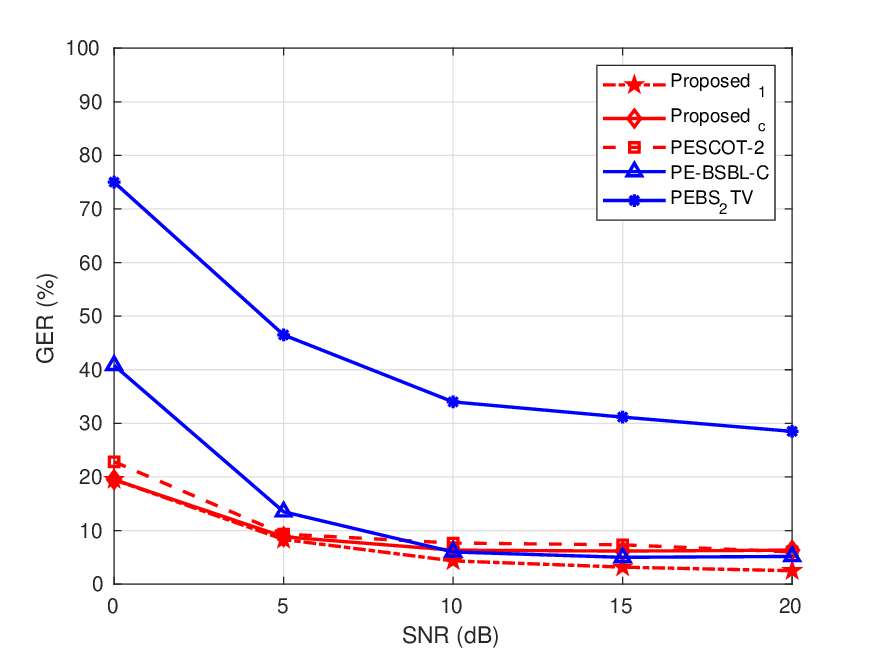}
    \caption{GER for simulated data, 4 pitches, with varying SNR.}  
    \label{fig:ger_SNR}\vspace{-1mm}
\end{figure}
We first conduct a Monte Carlo simulation study in which each signal consists of $4$ pitches, sampled at $8000$ Hz and observed over $250$ samples. The number of harmonics per pitch is randomly selected as an integer between $3$ and $10$, the magnitudes are set to $1$ and the initial phase is drawn uniformly at random, similar to what was done in \cite{shi2018multipitch, ELVANDER201656}. The \textit{nominal} fundamental frequencies are $176, 197, 240$ and $272$ Hz, respectively, and are perturbed in each simulation to avoid bias from grid effects. Unless otherwise stated, the signal-to-noise ratio (SNR) is set to $5$ dB, with the SNR defined as $\mathrm{SNR} = 10\log_{10}\left( (\sum_{k=1}^{K}\sum_{l=1}^{L^{(k)}}|\alpha_{l}^{(k)}|^2) / \sigma^2 \right)$, and all harmonics up to maximum harmonic order are non-zero. In order to evaluate the pitch estimation with respect to inharmonicity, we let the inharmonicity parameter $\kappa$ determine how much the harmonics of the \textit{nominal} fundamental frequency are allowed to deviate from the harmonic model. The harmonics $\omega$ are drawn uniformly at random in the interval $\omega \pm \kappa\omega$, where $\kappa\in [0,1]$. Here, again unless otherwise stated, the inharmonicity is set to $0$. The performance is evaluated in terms of the gross error rate (GER), defined as the percentage of pitch estimates that deviate by more than $50$ cents\footnote{This corresponds to an error larger than a quarter note.} from the reference pitch, i.e.,
\begin{align*}
    \text{GER} = \frac{1}{K}\sum_{k=1}^K \left[\left|1200\cdot \log_2\left(\frac{\hat{\omega}_0^{(k)}}{\omega_0^{(k)}}\right) \right| > 50 \right],
\end{align*}
where $\hat{\omega}_k$ is the pitch estimate \cite{Cheveigne2002YIN,shi2018multipitch}. We here denote the model from our previous work in \cite{bjorkman2025robust} as PESCOT-$2$, where the $2$ comes from its similarity to the problem in \eqref{eq:otdist}, with the difference being that the constraint is squared. For PE-BSBL-C \cite{shi2018multipitch} and PEBS$_2$TV \cite{ADALBJORNSSON2015236}, the maximum number of assumed harmonics of the signal is set to $10$, while for the proposed methods as well as PESCOT-$2$, the order is unrestricted. The proposed methods are denoted as Proposed$_1$, corresponding to the estimator of the deterministic model, and Proposed$_{c}$, corresponding to the estimator of the stochastic model. For all methods, the signal is normalized as to have unit variance.

Figure~\ref{fig:ger_SNR} displays results for varying levels of noise in the perfectly harmonic case. As can be seen, all methods save PEBS$_2$TV perform similarly in the high SNR regime.
Figure~\ref{fig:ger_inharm} shows the impact of inharmonicity, where a fixed SNR$~=~5$ dB is used. Although the results seem to imply that all methods are similarly robust to inharmonicity, a more nuanced picture emerges when studying the structure of the spectral estimates themselves, i.e., $\bnu$ and $\abs{\bmu}$, and not only the identified fundamental frequencies. To quantify this, we use the Wasserstein-2 distance \cite{villani2009optimal} between estimated and ground truth spectra. In particular, the Wasserstein-2 distance is defined as the square-root of the minimal objective of \eqref{eq:basic_ot}, where $c(\omega_1,\omega_2) = (\omega_1-\omega_2)^2$. Using this distance offers an alternative to, e.g., $\ell_2$-norms, as it is smooth with respect to perturbations such as peak shifts or splitting \cite{georgiou2008metrics}. It should be stressed that this distance is not related to the transport problems used to define the proposed estimators. The upper subplot in Figure~\ref{fig:OT_freq_id} shows the cumulative distribution function (CDF) of the Wasserstein-2 distance computed over 150 simulated signals, with inharmonicity 2\% and SNR 5 dB, while the lower subplot shows the mean of these distances with varying inharmonicity. As may be noted, the OT-based methods yield spectral estimates concentrated closer to the ground truth as compared to the reference methods.
\begin{figure}[t]
    \centering
    \includegraphics[width = \linewidth]{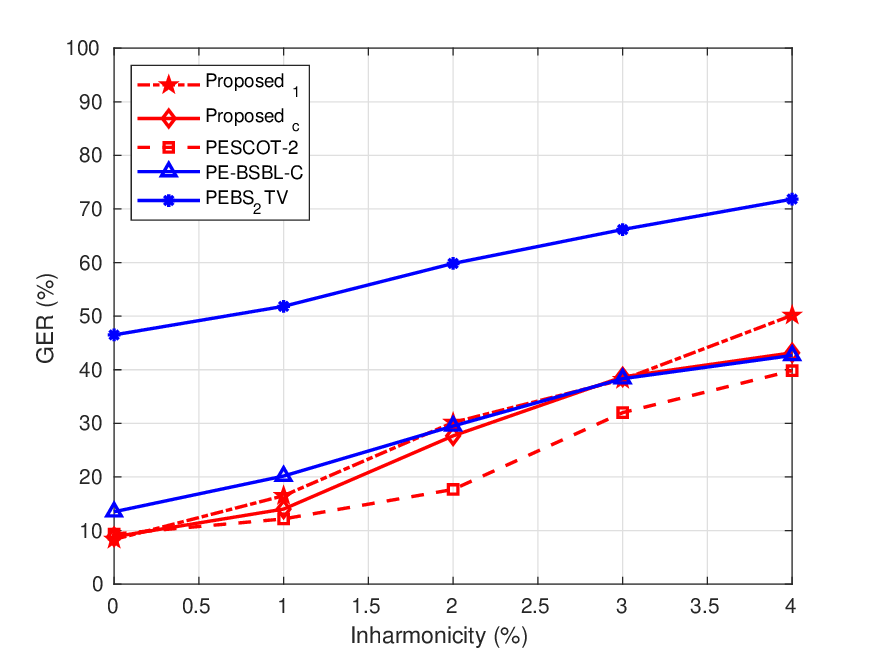}
    \caption{GER for simulated data, 4 pitches, with varying inharmonicity. The SNR is set to $5$ dB.}  
    \label{fig:ger_inharm}\vspace{-1mm}
\end{figure}
\begin{figure}[t]
    \centering
    \includegraphics[width = \linewidth]{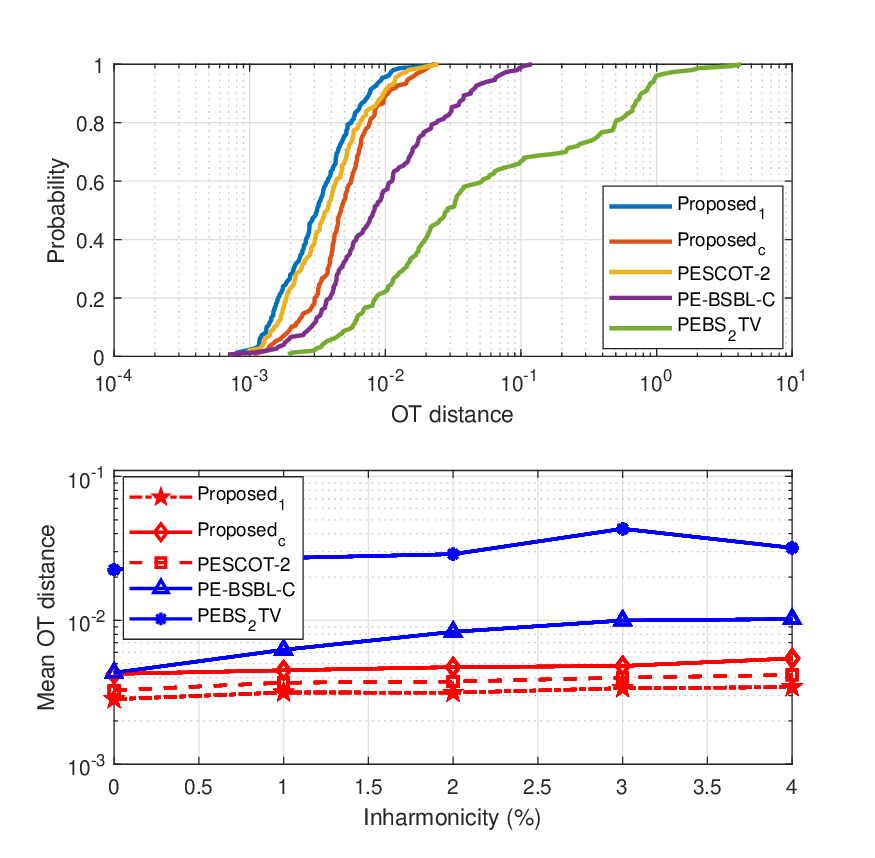}
    \caption{In the upper subplot, the empirical CDF of the OT distance between the estimated amplitude spectrum and the real amplitude spectrum is shown, with inharmonicity set to $2\%$. In the lower subplot, the median OT distance is shown as a function of the inharmonicity. The SNR for both plots is set to $5$, and the signals used are simulated with $4$ pitches.}  
    \label{fig:OT_freq_id}\vspace{-1mm}
\end{figure}
\begin{figure}[t]
    \centering
    \includegraphics[width = 0.95\linewidth]{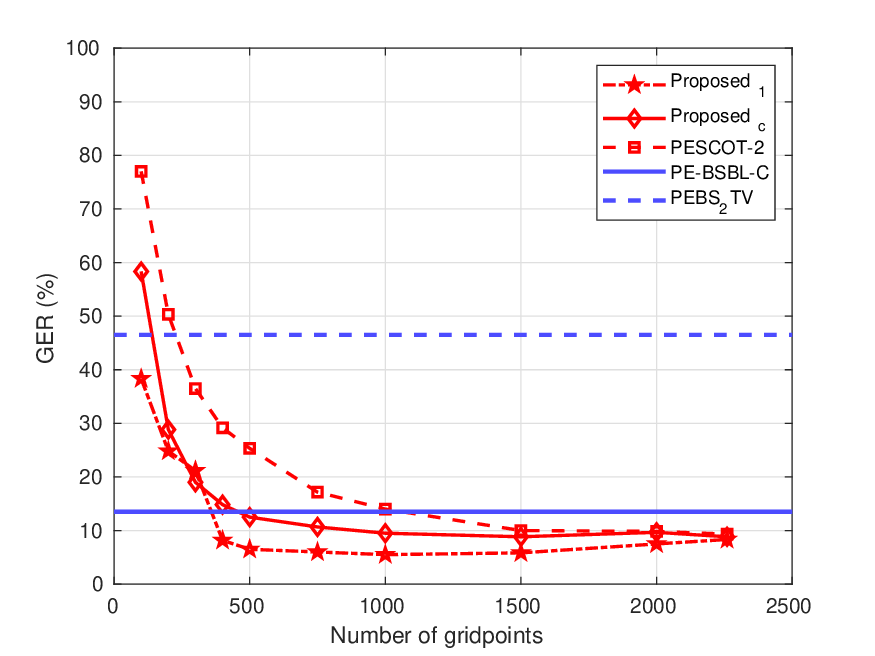}
    \caption{GER for simulated data, 4 pitches, with varying number of uniformly spaced grid-points from $50$ to $500$ Hz used by the proposed methods and PESCOT-2. The SNR is set to $5$, and the inharmonicity is set to $0$\%. }  
    \label{fig:ger_gridp}\vspace{-1mm}
\end{figure}
In order to evaluate the impact of the frequency grid resolution, in Figure \ref{fig:ger_gridp} we vary the total number of uniformly spaced grid points for the proposed methods and PESCOT-$2$. For all other experiments, this number is set to match the number of grid points for the reference methods, $2260$ grid points. The frequency grid for the reference methods is based on integer multiples, up to the maximum allowed number of harmonics, of the pitch grid, while reference methods, and PESCOT-$2$, have the flexibility of choosing any frequencies for its grid, allowing for reduced complexity for each iteration. As shown, somewhere between $1000$ and $1500$ uniformly spaced grid points appear to be sufficient for the experiment, with most notably Proposed$_1$ being able to achieve a similar performance with around $400$ grid points as it does with the full $2260$ grid points, greatly reducing the complexity of the iterations in the simulated setting. This concludes all the experiments on simulated data, and in the following section, we move on to evaluating the estimators on real data.
\subsection{Real data}
\begin{figure}[t]
    \centering
    \includegraphics[width = \linewidth]{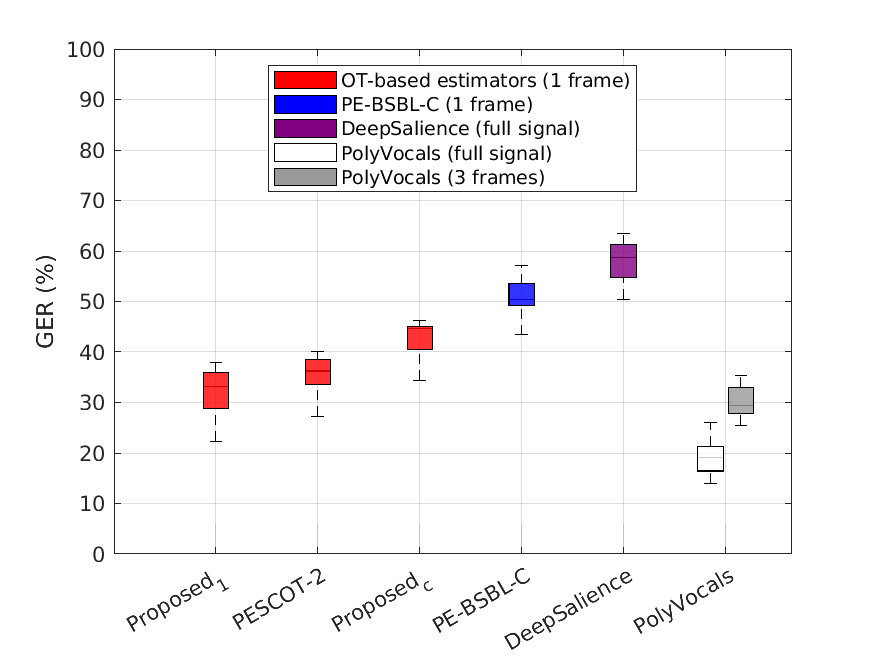}
    \caption{Boxplot of the GER for the Bach10 data. The proposed methods as well as PESCOT-$2$ and \text{PE-BSBL-C} use one frame per estimate, while the estimates in purple and white use the entire signal. The estimates in gray use 3 frames per estimate.}  
    \label{fig:ger_bach_full}\vspace{-1mm}
\end{figure}

The next experiment is based on the Bach10 dataset used in \cite{duan2010multiple}, where the pitch of 10 music pieces is evaluated. For each piece, the sampling rate is $44100$ Hz. Here, the ground truth fundamental frequencies were obtained through the single-pitch YIN estimator, which was used for each channel of the recording, corresponding to different instruments, separately, and was manually corrected in the case of obvious errors. The maximum number of assumed harmonics of the signal is set to $10$ for PE-BSBL-C, while for the network-based methods DeepSalience \cite{bittner2017deep} and PolyVocals \cite{cuesta2020multiple}, they were both trained on a maximum harmonic order of 5. For PolyVocals, there are several models available, where the one used here is the recommended model denoted as Deep/Late in the paper. For the proposed methods and PESCOT-$2$, the maximum order is unrestricted.

As is shown in Figure \ref{fig:ger_bach_full}, all the proposed methods outperform the reference methods except PolyVocals. For the non-network methods, the signal is divided into $30$ms long frames of music ($1323$ samples), where each frame is evaluated independently of the others. Although this frame-wise setup can also be used with the network-based methods, they failed to produce meaningful estimates with such short input signals and were therefore not added to the figure. However, for PolyVocals, we also evaluate the performance when the entire input signal is split into non-overlapping sequences of \textit{three} consecutive frames, i.e., $90$ms long segments (3969 samples). The result for this is shown in the gray box. For the white box, the entire signal is used. With this we note that PolyVocals does in fact outperform all the proposed methods; however, this is achieved by utilizing three times as much information as the proposed methods.\\
\\
The hyperparameters for the reference methods were set according to the recommendations in their respective papers. The hyperparameters for the proposed methods and PESCOT-$2$ were chosen through empirical tuning, and can be seen in Table~\ref{tab:sim_data} and Table~\ref{tab:real_data}, for the simulated and real data, respectively. The number of delays used for the auto-covariance was set to $\Tlag = \frac{2}{3}N$ for all data.

\captionsetup[table]{skip=4pt}
\begin{table}[t]
\centering
\begin{tabular}{ |p{1.4cm}|p{1.4cm}|p{1.4cm}|p{2cm}|  }
 \hline
 \multicolumn{4}{|c|}{Simulated data} \\
 \hline
 Parameter & Proposed$_1$ & PESCOT-$2$ & Proposed$_{\text{c}}$ \\
 \hline
 $\eta$     & $3\cdot10^{-2}$   & $1 \cdot 10^{0}$   &   $1\cdot 10^{-1}$ \rule{0pt}{8pt}\\
 $\zeta$    & $6\cdot 10^{0}$   & $8\cdot10^{3}$   & $1\cdot10^1$ \\
 $\epsilon$ & $1\cdot10^{-5}$   & $1\cdot 10^{-7}$   & $1\cdot10^{-6}$\\
 $\beta$    & $8\cdot 10^{-2}$   & $8\cdot 10^{-2}$ &  $1.5\cdot10^{-2}$\\
 \hline 
 \multicolumn{4}{|c|}{Debiased estimate} \\
 \hline
 $\zeta$    & $6\cdot 10^{0}$   & $8\cdot10^{3}$   & $1\cdot10^0$ \\
 $\beta$    & $4\cdot 10^{1}$   & $4.8\cdot 10^{1}$ &  $4.5\cdot10^{-2}$\\
 \hline
\end{tabular}
\caption{Parameter values for all proposed methods, for all the simulated data.}
\label{tab:sim_data}
\end{table}

\begin{table}[t]
\centering
\begin{tabular}{ |p{1.4cm}|p{1.4cm}|p{1.4cm}|p{2cm}|  }
 \hline
 \multicolumn{4}{|c|}{Real data} \\
 \hline
 Parameter & Proposed$_1$ & PESCOT-$2$ & Proposed$_{\text{c}}$ \\
 \hline
 $\eta$     & $5\cdot10^{-3}$   & $4\cdot 10^{-2}$   &   $1\cdot 10^{-2}$ \rule{0pt}{8pt}\\
 $\zeta$    & $1\cdot 10^{-1}$   & $8\cdot10^{0}$   & $1\cdot10^{-1}$ \\
 $\epsilon$ & $1\cdot10^{-5}$   & $1\cdot 10^{-7}$   & $1\cdot10^{-5}$\\
 $\beta$    & $2.8 \cdot10^{-2}$   & $3.6 \cdot 10^{-2}$&  $3.4\cdot10^{-2}$\\
 \hline 
\end{tabular}
\caption{Parameter values for all proposed methods, for the real data.}
\label{tab:real_data}
\end{table}

\section{Conclusion}

This work addresses the problem of estimating pitch in multi-pitch signals. We show that the proposed methods achieves state-of-the-art performance on both simulated and real data in a frame-based setting setting, while, unlike existing non-network-based methods, it does not require prior knowledge of the harmonic order of the signal. The performance is surpassed only by network-based estimators. It should, however, be noted that those rely on access to a larger portion of the signal, while our approach operates in a setting where only a short portion of a signal is measured, not utilizing any previous information of the signal, and is better viewed as a building block for a recursive estimator, utilizing all or some of the prior measurements that are available.
\section{appendix}
\vspace{-1mm}
\begin{proof}[Proof of Proposition~\ref{prop:prox_op_bregman}]
Introducing the auxiliary variable $\bQ = \RR^{F \times G }$, the minimization problem can equivalently be stated as
\begin{equation*}
\begin{aligned}
        \minwrt[\bnu,  \bM, \bQ] \quad & \stepsize\langle \bu, \bnu \rangle + \stepsize\zeta\langle \bC, \bM \rangle + \stepsize\zeta\epsilon D(\bM) \\
        &+\stepsize\zeta\eta\|\bQ\|_{\infty,1} + \mathcal{D}_{KL}(\bnu, \bnu^{(j)}) \\
    \text{s.t. } \quad & \bnu = \bM\onevec_F, \quad \bQ = \bM.
\end{aligned}
\end{equation*}
The corresponding Lagrangian (scaled by $1/\stepsize\zeta$) is
\begin{align*}
        &\mathcal{L}(\bnu, \bM, \bQ, \blambda, \bPsi) = \\
         &\frac{1}{\zeta}\langle \bu, \bnu \rangle + \left\langle \bC, \bM \right\rangle + \epsilon D(\bM) + \eta \| \bQ \|_{\infty,1}\\
        & + \frac{1}{\stepsize\zeta}\mathcal{D}_{KL}(\bnu, \bnu^{(j)}) + \left\langle  \blambda, \bnu - \bM \onevec_F \right\rangle + \left\langle \bPsi, \bQ - \bM \right\rangle ,
    \end{align*}
where $\blambda$ and $\bPsi$ are (scaled) dual variables.
It may be readily verified (see \cite[Thm. 1]{Haasler2024}) that the infimum of $\mathcal{L}$ respect to $\bQ$ is finite if and only if $\norm{\bPsi}_{1,\infty} \leq \eta$, in which case the terms related to $\bQ$ vanish. Furthermore, it may be noted that $\mathcal{L}$ is strictly convex with respect to $\bM$ and $\bnu$. The expressions \eqref{eq:optimal_nu} and \eqref{eq:optimal_stoch_M} are obtained by differentiating with respect to $\bnu$ and $\bM$ and solving the corresponding zero-gradient equations. The dual problem in \eqref{eq:prox_dual_problem_cov} is obtained by plugging these expressions into $\mathcal{L}$.
\end{proof}
%
%%%% Maximum entropy proof
\begin{proof}[Proof of Proposition~\ref{prop:maximum_entropy}]
Define the set
\begin{align*}
    U_\rho = \{ \bM \in \RR_+^{F\times G} \mid \norm{ \hat{\br} - \bA(\bM \onevec_G)}_2^2 \leq \rho  \},
\end{align*}
parametrized by $\rho > 0$. As $\bM$ is constrained to the positive quadrant and as the first row of $\bA$ is all ones (corresponding to $\tau = 0$), the set $U_\rho$ is bounded. By the continuity of $\bA$ and the norm $\norm{\cdot}$, $U_\rho$ is closed. Thus, $U_\rho$ is compact. Consider the problem
\begin{align} \label{eq:prop_proof_problem}
    \minwrt[\bM \in U_p] \quad f(\bM) + \epsilon D(\bM) 
\end{align}
where
\begin{align*}
    f(\bM) = \langle \bC, \bM \rangle + \beta \onevec_F^T\bM\onevec_G + \eta \norm{\bM}_{\infty,1}.
\end{align*}
Let $\{ \epsilon_p \}_p$ be a sequence going to zero, and let $\{ \bM_p \}_p$ be the corresponding sequence of solutions to \eqref{eq:prop_proof_problem}. Then, as $U_p$ is compact, there exists a convergent subsequence $\{\bM_q\}_q$ converging to some $\bM^\star \in U_\rho$. Let $\bM_o$ be a solution to \eqref{eq:prop_proof_problem} for $\epsilon = 0$. Then, for any $\epsilon_q$,
\begin{align*}
    f(\bM_q) + \epsilon_q D(\bM_q)  \leq f(\bM_o) + \epsilon_q D(\bM_o).
\end{align*}
Furthermore, $f(\bM_o) \leq f(\bM_q)$. Taken together,
\begin{align*}
    0 \leq f(\bM_q) - f(\bM_o) \leq \epsilon_q\left( D(\bM_o) - D(\bM_q) \right).
\end{align*}
As $f$ and $D$ are continuous, we by letting $q\to \infty$ get that the right-hand side goes to zero and $\lim_{q\to\infty}f(\bM_q) = f(\lim_{q\to\infty} \bM_q) = f(\bM^\star) = f(\bM_o)$. Thus, $\bM^\star$ solves \eqref{eq:prop_proof_problem} with $\epsilon = 0$. Furthermore, we have directly that $D(\bM^\star) \geq D(\bM_o)$ for any non-regularized solution $\bM_o$. By the strict convexity of $D$, uniqueness of this maximum-entropy solution follows. Lastly, for any $\zeta > 0$, there exists a corresponding $\rho > 0$ defining $U_\rho$, concluding the proof.
\end{proof}

\bibliographystyle{./IEEEtran}
\bibliography{./IEEEabrv,./refs}
\end{document}